\documentclass[journal]{IEEEtran}

\usepackage{subfigure}
\usepackage{pdfpages}
\usepackage{diagbox}
\usepackage{colortbl}
\definecolor{mygray}{gray}{0.85}
\usepackage{array}
\usepackage{caption}
\usepackage[linesnumbered,boxed,ruled,commentsnumbered]{algorithm2e}
 \usepackage{enumerate}
        \usepackage{setspace}

\usepackage{amsmath,bm}
\usepackage{cite}
\usepackage{amsmath,amssymb,amsfonts}
\usepackage{algorithmic}
\usepackage{amsthm}
\usepackage{graphicx}
\usepackage{textcomp}
\usepackage{xcolor}

\newtheorem{theorem}{Theorem}

\usepackage{hyperref}
\hypersetup{%
  bookmarks=true
}

\begin{document}

\title{
Federated Spectrum Learning for Reconfigurable Intelligent Surfaces-Aided Wireless Edge Networks
}

\author{
\IEEEauthorblockN{Bo~Yang,~
 Xuelin Cao, Chongwen Huang,~
  Chau Yuen, \IEEEmembership{Fellow,~IEEE},  Marco Di Renzo, \IEEEmembership{Fellow,~IEEE},  Yong Liang Guan,~\IEEEmembership{Senior Member,~IEEE},  Dusit Niyato,~\IEEEmembership{Fellow,~IEEE},  Lijun Qian, \IEEEmembership{Senior Member,~IEEE}, and M\'erouane Debbah, \IEEEmembership{Fellow,~IEEE}
  }
 
%
%

\thanks{

B. Yang is with the School of Computer Science, Northwestern Polytechnical University, Xi'an, Shaanxi, 710129, China (email: yang$\_$bo@nwpu.edu.cn)

 X. Cao is with the School of Cyber Engineering, Xidian University, Xian
710071, China, and also with the Engineering Product Development Pillar,
Singapore University of Technology and Design, Singapore 487372 (email: caoxuelin@xidian.edu.cn).  
 
 C. Yuen is with the Engineering Product Development Pillar, Singapore University of Technology and Design, Singapore 487372 (email: yuenchau@sutd.edu.sg). 

C. Huang is with College of Information Science and Electronic Engineering, Zhejiang University, Hangzhou 310027, China, and with International Joint Innovation Center, Zhejiang University, Haining 314400, China, and also with Zhejiang-Singapore Innovation and AI Joint Research Lab and Zhejiang Provincial Key Laboratory of Info. Proc., Commun. \& Netw. (IPCAN), Hangzhou 310027, China. (E-mail: chongwenhuang@zju.edu.cn ).

M. Di Renzo is with Universit\'e Paris-Saclay, CNRS, CentraleSup\'elec, Laboratoire des Signaux et Syst\`emes, 3 Rue Joliot-Curie, 91192 Gif-sur-Yvette, France (email: marco.di-renzo@universite-paris-saclay.fr). 

Y. L. Guan is with the School of Electrical and Electronic Engineering, Nanyang Technological University, Singapore (e-mail: eylguan@ntu.edu.sg).

 D. Niyato is with the School of Computer Science and Engineering, Nanyang Technological University, Singapore (email: dniyato@ntu.edu.sg) .

 L. Qian is with the Department of Electrical and Computer Engineering and CREDIT Center, Prairie View A$\&$M University, TX 77446, USA (email: liqian@pvamu.edu). 
 
M. Debbah is with the Technology Innovation Institute, 9639 Masdar City, Abu Dhabi, United Arab Emirates (email: merouane.debbah@tii.ae) and also with CentraleSupelec, University Paris-Saclay, 91192 Gif-sur-Yvette, France.

}

 }

\maketitle
\vspace{-10mm}

\begin{abstract}
Increasing concerns on intelligent spectrum sensing call for efficient training and inference technologies. In this paper, we propose a novel federated learning (FL) framework, dubbed {federated spectrum learning (FSL)}, which exploits the benefits of reconfigurable intelligent surfaces (RISs) and overcomes the unfavorable impact of deep fading channels. \textcolor{black}{Distinguishingly, we endow conventional RISs with spectrum learning capabilities by leveraging a fully-trained convolutional neural network (CNN) model at each RIS controller, thereby helping the base station to cooperatively infer the users who request to participate in FL at the beginning of each training iteration.} To fully exploit the potential of FL and RISs, we address three technical challenges: RISs phase shifts configuration, user-RIS association, and wireless bandwidth allocation. The resulting joint learning, wireless resource allocation, and user-RIS association design is formulated as an optimization problem whose objective is to maximize the system utility while considering the impact of FL prediction accuracy. In this context, the accuracy of FL prediction interplays with the performance of resource optimization.
In particular, if the accuracy of the trained CNN model deteriorates, the performance of resource allocation worsens. 
The proposed FSL framework is tested by using real radio frequency (RF) traces and numerical results demonstrate its advantages in terms of spectrum prediction accuracy and system utility: a better CNN prediction accuracy and FL system utility can be achieved with a larger number of RISs and reflecting elements.



\end{abstract}



\section{Introduction}

\IEEEPARstart{S}{ixth-generation} (6G) networks are envisioned to provide new services and applications to the users. Due to the limited spectrum resources, spectrum sensing, which is defined as the task of ascertaining the spectrum usage and the activity of mobile users, has become urgent and meaningful to improve the spectrum usage efficiency and address the spectrum scarcity problem in heterogeneous wireless networks~\cite{ss_01, ss_02}. For instance, the desire for better spectrum utilization has triggered spectrum sharing such as the coexistence of WiFi and Long Term Evolution (LTE) in unlicensed spectrum (LTE-U)~\cite{cxl}. 
 However, large-scale dynamical networks
lead to spectrum characteristics uncertainty, which entails great difficulty in robust and accurate spectrum sensing. To address this challenge,  deep learning (DL) has been employed to intelligently identify the spectrum characteristics~\cite{survey_sensing,deep_01}. Specifically, learning deep features from the radio frequency (RF) signals via appropriately offline trained deep neural networks (DNNs) has become popular~\cite{TWC_YB,WCL_sensing,Access_YB,TMC_sensing}.  

It is worth noting that, in traditional DL-based spectrum sensing methods, pre-trained DNNs that detect RF signals usually need to be fine-tuned or even re-trained once the spectrum occupancy changes significantly. As for the training of DNNs, conventional centralized training approaches (e.g., training at a central cloud) that require uploading a large amount of raw data may not be feasible in practice due to the limited communication bandwidth, data privacy and security concerns. In this case, federated learning (FL) has emerged as a suitable solution to leverage personalized dataset from a large number of mobile users in order to collaboratively train a shared DNN in a decentralized way, while achieving differential data locality~\cite{FL_google,DL_edge}. In each iteration of FL (called round), the users train their local models based on their own data and then upload, via wireless links, the model updates to an edge server where the global aggregation is performed.

\subsection{Motivation and Scope}
In the depicted context, several new research challenges emerge for the conventional FL process:
\begin{itemize}
\item \textit{Challenge-1:} Before the start of local training at each iteration of the conventional FL process, the edge server first needs to perform resource allocation (such as users selection and bandwidth allocation) by solving the corresponding optimization problem. Due to the existence of shadowing and multi-path fading in radio environments, some users may not be able to reliably send their requests to the base station (BS) via the direct links, which may be severely degraded and then unreliable. In this case, the BS cannot perform the resource allocation well since some of the requests may be missed.

\item \textit{Challenge-2:} In conventional FL, the users who finish the local training send via wireless links their local models to the BS for global model aggregation. Due to the possible unreliability of the wireless direct links and the limited wireless bandwidth, the convergence of FL may be degraded because of the low efficiency of uploading the local models. In this case, the accuracy of FL prediction may be degraded since some of the local models may be received unsuccessfully at the BS. 
\end{itemize}

To address the aforementioned challenges in conventional FL implementations, emerging and advanced technologies such as terahertz (THz) communications and ultra-massive multiple-input-multiple-output (MIMO), can be exploited to improve the wireless transmission. However, excessive hardware costs and energy consumption pose design challenges due to the relatively high operating frequency of THz communications and the large number of RF chains required. In recent years, reconfigurable intelligent surfaces (RISs), benefited from breakthroughs on the fabrication of programmable metasurfaces, have emerged as a promising technology for improving the quality of wireless links and for configuring the wireless environment by, e.g., appropriately reflecting the incident signals with the aid of a large number of {nearly-passive reflecting} elements 
with low power consumption~\cite{EURASIP, JSAC_ris_02, {Huang02},{WU01},{jsac}}. Recently, innovative RIS technologies improving wireless communications with low cost and energy consumption have been proposed~\cite{{THz_01},{JSAC_ris_01}}. For example, holographic MIMO surfaces (HMIMOS) attracted great research attention as a possible solution for realizing massive MIMO systems at a lower cost and power consumption~\cite{{holo_01},{holo_02}}. 
In addition, to collect users' training requests subject to deep fading conditions, cooperative RF spectrum sensing constitutes a promising approach for improving the sensing performance by capitalizing on the spatial sensing diversity of distributed sensors~\cite{cooperative_01,cooperative_02,cooperative_03}. 
\textcolor{black}{Motivated by these considerations, it is desirable to address the two mentioned challenges of FL by integrating FL with RISs, thereby achieving cooperative RF spectrum sensing (to improve the performance of optimization) and enhancing the quality of the wireless links (to improve the accuracy of FL prediction)~\cite{FL_RIS_01}. }

%


\vspace{-3mm}
\subsection{Related Work} 
Recently, some researchers have investigated the potential benefits of using RISs for improving the performance of FL~\cite{{RIS_FL_03},{RIS_FL_04}}. In these papers, however, the considered system scenarios are limited to a single-RIS setup, 
 and the user-RIS association problem was not relevant. 
A more complex and general FL system has been investigated in~\cite{RIS_FL_01}, where multiple RISs are deployed for performance enhancement. However, none of the above existing works explored the interplay between FL and RISs. Thus, this motivates us to design a novel federated spectrum learning (FSL) framework, thereby achieving efficient spectrum learning for wireless edge networks.

In wireless networks, the potential opportunities offered by the distributed FL paradigm have not been fully exploited, mainly because of the {straggler effect} and the unreliability of wireless channels under limited wireless resources~\cite{{FL_chen_01}}. 
Due to the unreliability of wireless channels, in particular, it may be difficult to receive correct local updates at the edge, especially when the wireless channel undergoes deep fading. As a result, the convergence performance of FL may be degraded. In this context, it is imperative to improve the FL performance from the communications perspective~\cite{{FL_chen_01},{TOC_FL},{straggler_03},{straggler_04}}. Specifically, Chen \textit{et al.}~\cite{FL_chen_01} investigated a joint learning, wireless resource allocation, and
user selection problem in FL by taking into account the wireless link packet errors and the availability of wireless resources. 
 Samarakoon \textit{et al.}~\cite{TOC_FL} introduced FL to help with the joint design of power control and resource allocation for application to vehicular communication networks. To mitigate the straggler effect, 
  Yang \textit{et al.}~\cite{straggler_03} proposed a fast global model aggregation approach to improve the performance and the convergence rate of FL via over-the-air computation (AirComp). Since uploading of the local parameters at every iteration is often inefficient in FL, Wang \textit{et al.}~\cite{straggler_04} presented a control algorithm that determines the best tradeoff between local update and global aggregation so as to minimize the loss function of FL.

  Unlike other transmission technologies such as active relay, RISs can help to reflect the incident signals through adjusting the phase shifts of their scattering elements smartly, thereby improving the wireless transmissions~\cite{add_cxl}. 
In order to explore the benefits of RISs in wireless networks, designing RIS-aided wireless communication systems based on deep learning techniques has recently received major attention~\cite{cao_magazine}. 
 In particular,  Hu \textit{et al.}~\cite{MetaSensing} considered a metasurface assisted RF sensing method which can sense the locations of objects in a 3D space, and proposed a deep reinforcement learning (DRL) algorithm to solve the formulated optimization problem.  Huang \textit{et al.}~\cite{RIS_RL_02} analyzed the joint design of transmit beamforming at the base station and phase shifts at the RIS to maximize the sum rate 
utilizing DRL. Cao \textit{et al.}~\cite{JSAC_cao} investigated an RIS-assisted multi-user downlink aerial-terrestrial communication system via multi-task learning.
 Yang \textit{et al.}~\cite{RIS_RL_01} proposed an RIS-assisted anti-jamming solution for securing wireless communications via RL. 
  Although some works focused on deep learning for RIS systems, e.g., reflecting beamforming matrix optimization~\cite{MetaSensing,{RIS_RL_01},{RIS_RL_02},JSAC_cao} and channel estimation~\cite{{DL_RIS_01},{DL_RIS_02},{DL_RIS_03}}, {endowing RISs with {active learning} capabilities for data-driven wireless networks still presents many open issues to be explored and addressed.}

  \vspace{-3mm}

\subsection{Contributions and Organizations}
 
In this paper, 
our goal is to propose a novel FSL framework to improve the accuracy of FL  with the aid of RISs, as well as to endow conventional RISs with spectrum sensing capabilities, thereby improving the system performance in a cooperative way, especially in radio environments subject to deep fading channels. In this context, our major contributions are summarized as follows:
\begin{enumerate}
\item 
{We develop a novel FSL framework for RISs-aided wireless edge networks by investigating the interplay between FL and RISs.} Based on collected RF data traces, the mobile users train their local models and then upload them to an edge server via RIS-assisted links, thereby achieving a fast yet reliable model aggregation.
At the beginning of each iteration, in addition, the aggregated model is deployed at each RIS controller to help the BS collect the requests of users by cooperatively sensing the spectrum, thereby improving the performance, especially in complex shadowing and fading radio environments. 



\item 
{We jointly optimize the user-RIS association, the phase shifts configuration and the wireless bandwidth allocation.}
In the considered FSL framework subject to limited wireless resources and a limited number of RISs, specifically, we formulate a joint user-RIS association, phase shifts design, and wireless bandwidth allocation problem, whose goal is to maximize the FL system utility. 
To solve the user-RIS association problem, we propose a matching game-based association scheme in which the users who achieve larger gains in terms of achievable SNR have higher probabilities to be associated with an RIS. The optimized phase shifts guarantee that the achieved gain of each user served by the RISs is maximized. Also, we employ the bisection search method to handle the wireless bandwidth allocation problem. 

\item {We demonstrate the advantages of the proposed FSL framework in terms of the training and inference performance.} The proposed FSL framework is tested by using real RF traces, and numerical simulations demonstrate that the proposed FSL framework outperforms other benchmark schemes regarding spectrum sensing accuracy and system utility. \textcolor{black}{Based on the conducted study, the following considerations can be made: 
\vspace{-1.0mm}
\begin{itemize}
\item The proposed FSL framework assisted by a larger number of RISs can achieve a better prediction accuracy and lower training loss. \vspace{-0mm}
\item As the number of reflecting elements of each RIS increases, the FL system utility can be significantly improved. \vspace{-0mm}
\item The performance of the proposed FSL framework is determined by the learning rate in the presence of multiple RISs. Therefore, an appropriate learning rate needs to be chosen as a function of the number of available RISs.
\end{itemize} } 
\end{enumerate} 

\textit{To our best knowledge, this is an early work that investigates the interplay between FL and RISs for efficient spectrum learning over wireless edge networks.}
The rest of this paper is organized as follows. We describe the proposed FSL framework in Section~\ref{FSL framework}. We introduce the system model and problem formulation in Section~\ref{system}. We present the algorithm solution to the considered optimization problem in Section~\ref{soving}. We illustrate the simulation results in Section~\ref{result}. Finally, conclusions are drawn in Section~\ref{conclusion}.

\textit{Notations:} A bold letter indicates a vector or matrix. An upper  case letter indicates a random variable or random parameter and a lower case letter indicates a realization of a random variable or a random parameter. $\max \{  \cdot \}$ and $\min \{  \cdot \}$ represent the maximum value and the minimum value, respectively. The amplitude of a complex number $a$ is denoted by $\left | a \right |$. The symbol $\leftarrow$ denotes the `assignment' relation. 

\begin{figure*}[t] 
\centering
\captionsetup{font={footnotesize }}
\subfigure[RIS-aided resource allocation stage]{
\includegraphics[width=3.3in,height=2.25in]{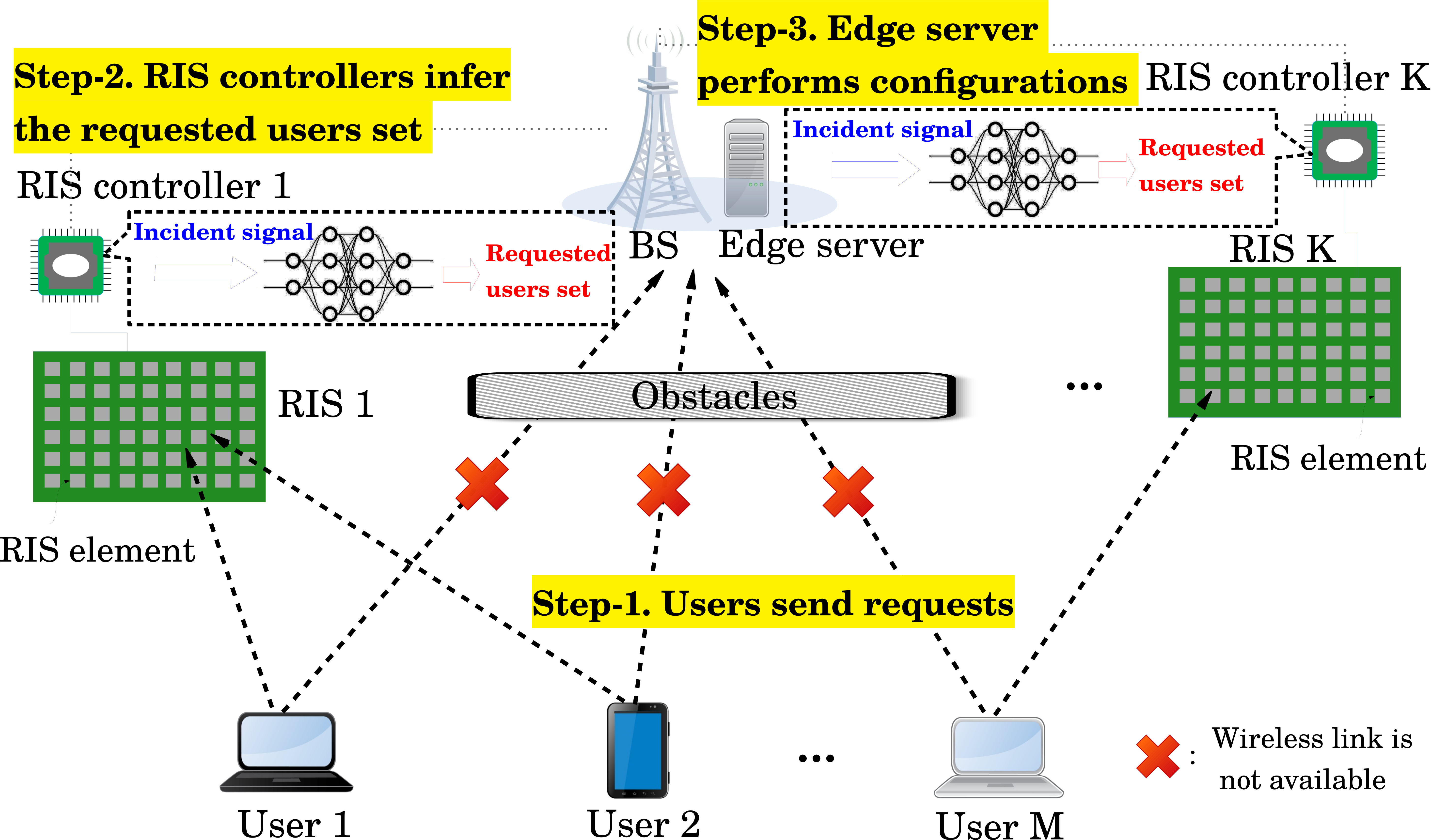}
}
\hspace{0.85mm}
\subfigure[RIS-aided FL training stage]{
\includegraphics[width=3.3in,height=2.25in]{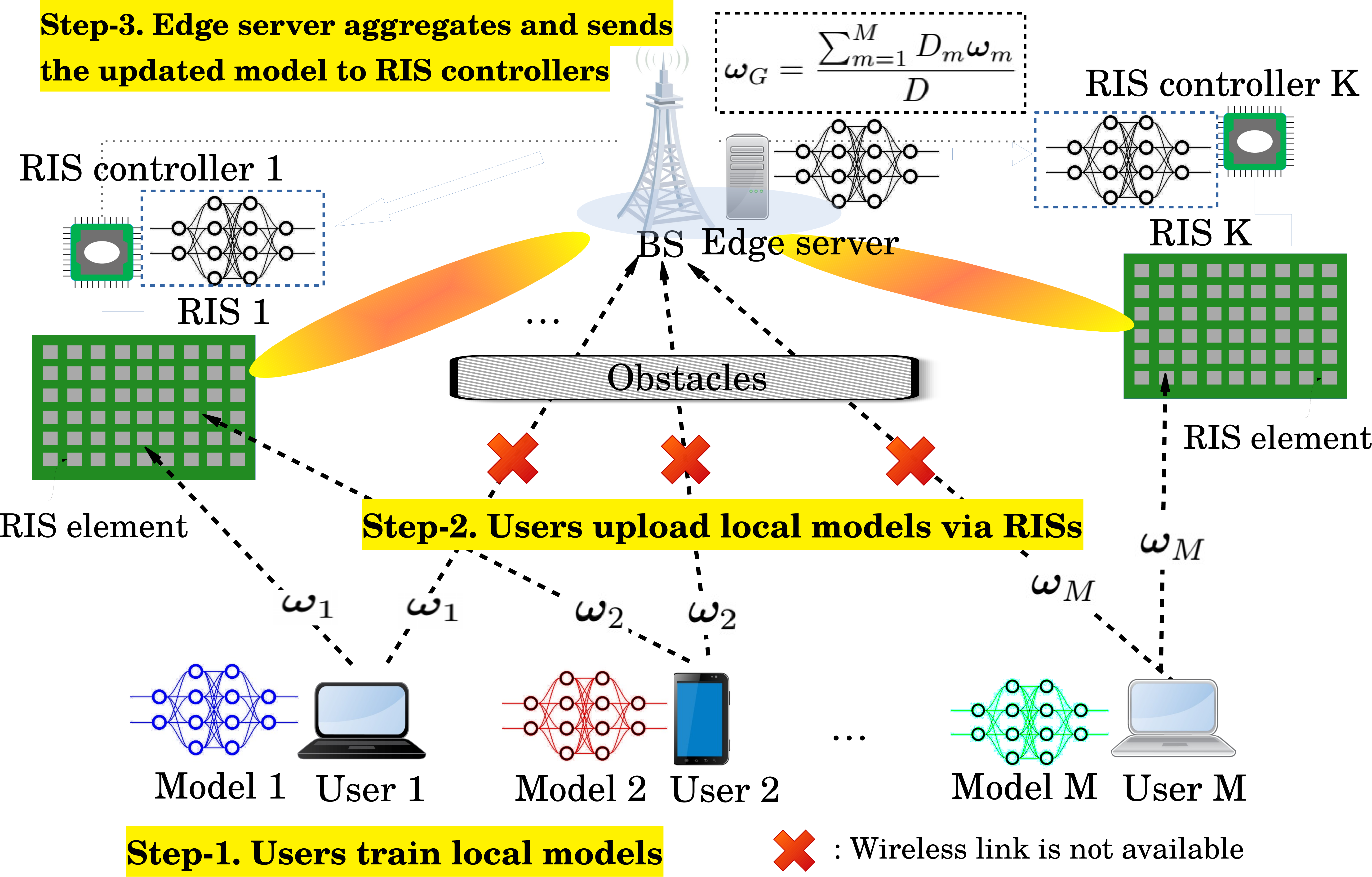}}
\caption{The proposed FSL framework encompasses two stages: the RIS-aided resource allocation stage (shown in (a)) and the RIS-aided training stage (shown in (b)). These two stages interplay with each other and constitute one iteration (or round) of FL, which is repeated several times until the FL model converges. In (a), the FL model is deployed at each RIS controller to infer the users who send requests (called requested users). In (b), RISs help to improve the wireless transmissions between the users and the base station.}
\label{system_model}
\end{figure*}

\section{Federated Spectrum Learning Framework}
\label{FSL framework}
\subsection{Network Scenario}

As shown in Fig.~\ref{system_model}, we consider a wireless network that consists of one base station (BS) co-located with an edge server, which serves a set $\cal M$  of $M$ users via $K$ RISs, where ${\cal M}=\{1,2,\ldots,M\}$. There may exist obstacles (e.g., buildings) between the users and the BS, and the RISs can help improve the quality of wireless transmission between the users and the BS. We denote the $m$th user as $U_m , \ m \in {\cal {M}}$, the local training dataset by ${\cal D}_m$ with $D_m=|{\cal D}_m|$ being the number of data samples.  In the dataset ${\cal D}_m=\{\bm{x}_{m,s}, {y}_{m,s}\}^{D_m}_{s=1}$, each data sample is constituted by an input vector $\bm{x}_{m,s} \in \mathbb{R}^{N_{d}\times 1}$ and its corresponding output value ${y}_{m,s}\in \mathbb{R}$.  As far as the data locality is concerned, we assume that the collected datasets are non-overlapping with each other, i.e., $\bm{x}_i \neq \bm{x}_j$, $i \neq j, \ \forall i, j \in \cal M$.

\vspace{-3mm}

\subsection{Design Basics} 
\textcolor{black}{In this paper, a novel FSL framework is proposed to exploit the interplay between FL and RISs. 
In the considered FL wireless network, specifically, only some users may send requests (called requested users) to participate in the training procedure of FL. Therefore, it is necessary to identify them and to appropriately allocate the available resources (e.g., wireless bandwidth) for maximizing the FL performance. In the proposed FSL framework, the requested users are inferred with the aid of neural networks deployed at the controller of the RISs. The latter neural networks are in turn trained by leveraging an FL framework that relies on the availability of RISs for reliable data transmission. These two phases are intertwined and, therefore, each training iteration of the proposed FSL framework consists of two stages: 1) the RIS-aided resource allocation stage, and 2) the RIS-aided FL training stage.}

\textcolor{black}{Specifically, in the RIS-aided resource allocation stage shown in Fig.~\ref{system_model}(a), an initial convolutional
neural network (CNN) model is deployed at each RIS controller to collaboratively infer the requested users from the impinging RF signals\footnote{We assume that the RIS controller has sufficient computation resources, e.g., field programmable gate
array (FPGA)~\cite{RIS_FPGA}, to run the trained model for inference.}, and then the obtained estimates are reported to the BS for resource allocation. In the RIS-aided federated training stage shown in Fig.~\ref{system_model}(b), given to the allocated resources, the CNN model available at the RISs' controllers are further trained with the aid of FL, where the requested users upload their local model parameters via RISs-aided wireless links. It is worth noting that FL training and resource allocation (or optimization) are two intertwined problems, as illustrated in Fig.~\ref{structure}.
In particular, if the resource allocation cannot yield an optimal solution, i.e., the wireless bandwidth is not appropriately allocated to the requested users, then the accuracy of the CNN model trained via FL worsens. If the accuracy of the trained CNN model deteriorates, in turn, the performance of resource allocation worsens. Indeed, this is a typical `chicken-and-egg' issue. To address this issue, we propose the FSL framework by leveraging a fully-trained CNN model at each RIS controller and by jointly optimizing the user-RIS association, the phase shifts configuration and the wireless bandwidth allocation.} 

 
\begin{figure}[t]
  \captionsetup{font={footnotesize}}
\centerline{ \includegraphics[width=3.25in, height=2.32in]{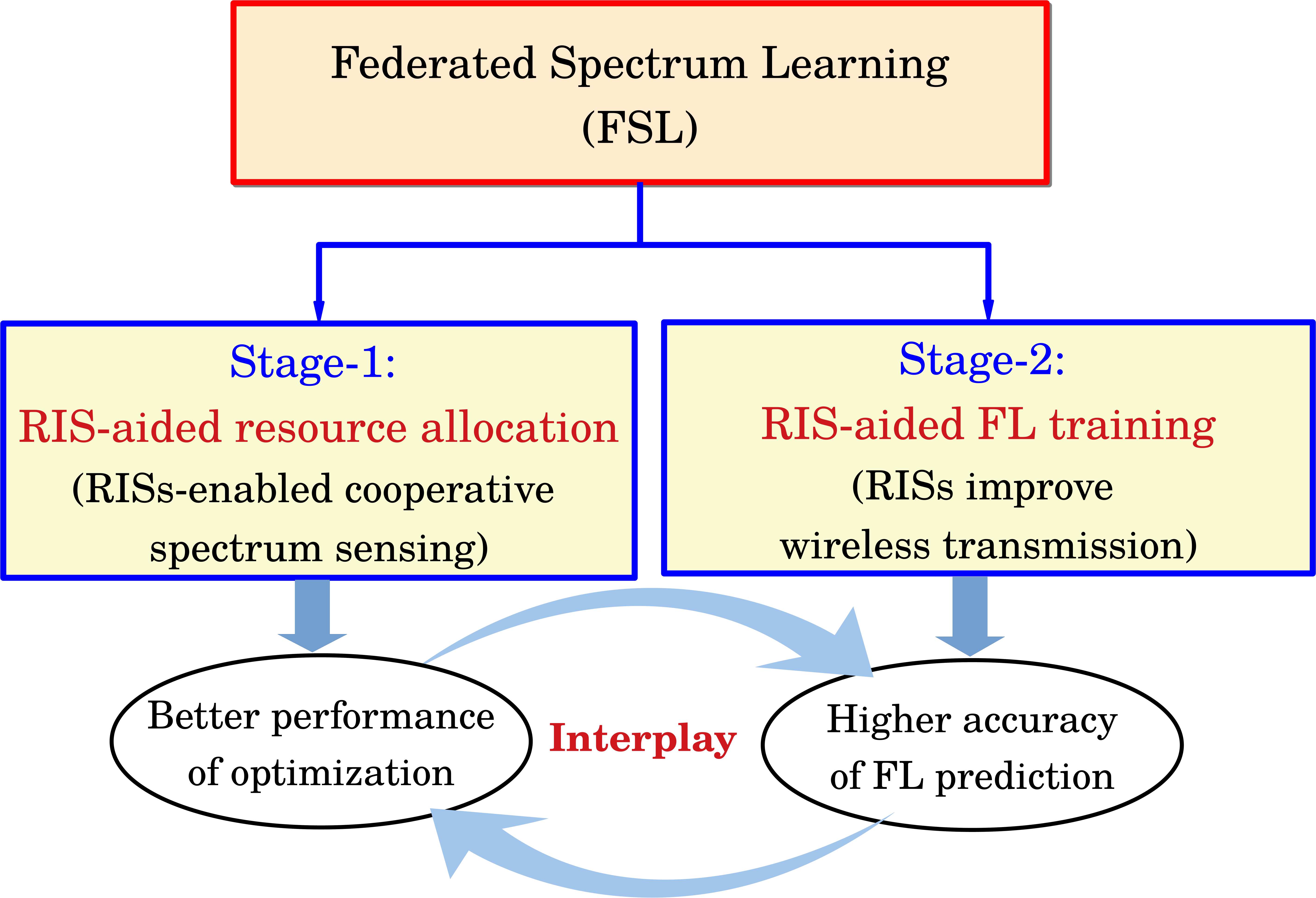}}
\vspace{-1mm}
\caption{The interplay between the two stages in the proposed FSL framework.}
\label{structure}
\end{figure}

 \theoremstyle{Remark}
\newtheorem{remark}{\textbf{Remark}}  
\theoremstyle{Assumption}
\newtheorem{assumption}{\textbf{Assumption}}

 \subsection{RIS-aided Resource Allocation Stage}
 As illustrated in Fig.~\ref{system_model}(a), the RIS-aided resource allocation stage involves three steps: {\textit{i}) Users send requests:} the users who aim to participate in the current iteration of FL send requests to the BS; {\textit{ii}) RIS controllers infer the set of requested users:} the CNN model deployed at each RIS controller estimates the ID of the requested users, which are then reported to the BS via a dedicated channel;  {\textit{iii}) Edge server performs configurations:} Based on the inferred results, the edge server performs the resource allocation by solving the optimization problem and broadcasts the result to the users and the RIS controllers.   
 \textcolor{black}{Specifically, when arriving at the RISs, the incident RF signals first undergo analog-to-digital conversion (ADC) and frequency down-conversion. Then the baseband In-phase (I) and Quadrature (Q) sequences are fed into the trained CNN model to perform online inference at the RIS controller.} 
\textcolor{black}{To achieve wireless spectrum learning at the RIS controllers, we introduce \textbf{Assumption~\ref{assumptionO1}}.}
\vspace{-4mm}

\textcolor{black}{\begin{assumption} \label{assumptionO1}
 To enable the RISs to identify the requested users according to the incident signal, we assume that each RIS is equipped with a few {`semi-active' RIS elements}, which obtain the I/Q sequences by performing ADC and down-conversion of the incident RF signals.
 \end{assumption}}
 
 By performing feed-forward calculation via the trained CNN model, the set of inferred requested users obtained by the $k$th RIS controller is represented by 
\begin{equation} \label{if_set} 
\widetilde{\cal U}_k = \begin{cases}
\{U_m\}, \ \forall m \in {\cal M},   & \text{if } \widetilde{n}_k \geq 1, \\ 
\varnothing , & \text{if } \widetilde{n}_k=0,
\end{cases}
\end{equation}
 where $\widetilde{n}_k$ indicates the inferred total number of requested users in the incident signal.
 
  
 \subsection{RIS-aided FL Training Stage}
In the considered FSL framework, for a data sample with input $\bm{x}_m$, the task is to find a model parameter vector $\bm{\omega}_m$ that characterizes the output $y_m$ by minimizing the loss function $f(\bm{\omega}_m, \bm{x}_{m,s}, {y}_{m,s})$. Since the dataset of $U_m$ is ${\cal D}_m$, the loss function of $U_m$ can be obtained as
\begin{equation} 
\ J_m(\bm{\omega}_m) = \frac{1}{D_m} \sum_{s=1}^{D_m}f(\bm{\omega}_m, \bm{x}_{m,s}, {y}_{m,s}).
\end{equation}

We denote $D=\sum_{m \in \cal M}D_m$ as the total data samples of all users, and the FL training process is performed to solve
the following optimization problem:
\begin{equation} 
\underset{\bm{\omega}_m}{\rm min} \ J(\bm{\omega}_m)  \triangleq \underset{\bm{\omega}_m}{\rm min} \left\{\frac{1}{D} \sum_{m=1}^{M}\sum_{s=1}^{D_m}f(\bm{\omega}_m, \bm{x}_{m,s}, {y}_{m,s}) \right\}.
\end{equation}


The FL training process is illustrated in Fig.~\ref{system_model}(b), which includes three steps: {\textit{i}) Users train local models:} each user $U_m$ trains the local model $\bm{\omega}_m$ based on the local dataset ${\cal D}_m$; {\textit{ii}) Users upload local models:} the users upload their trained local models (e.g., $\bm{\omega}_1, ..., \bm{\omega}_m$) via the RISs-aided wireless links to the edge server; {\textit{iii}) Model aggregation and broadcasting:} upon receiving all the uploaded gradients, the edge server aggregates the local updates and generates an updated model, $\bm{\omega}_G$, which is sent back to the requested users and the RIS controllers\footnote{Note that, different learning algorithms can be used to update the local FL model, e.g., gradient descent, to update the local FL model. The update of the global model $\bm{\omega}_G$ is given by $
\bm{\omega}_G = \frac{\sum_{m=1}^{M} D_m \bm{\omega}_m}{D}$.}. This training procedure is repeated several times until the global CNN model converges.

For each communication round, the uploading of the user's local FL model ($\bm{\omega}_m$) can be considered as a transmission that occurs in one time slot, which may not be recovered correctly at the edge server due to the impact of the wireless channel between the users and the edge server, especially when the wireless channel undergoes deep fading due to the presence of obstacles. In this case, the local FL model parameters cannot be used for the aggregation at the edge server. \textcolor{black}{To improve the quality of the global model that is aggregated at the edge server, the users can upload their local model parameters with the aid of RISs\footnote{\textcolor{black}{Note that the user-RIS association may need to be updated at the beginning of each FL round by using the proposed user-RIS association algorithm.}}.  }

Therefore, the aggregated FL model at a considered iteration in the proposed FSL framework is  
\begin{equation}\label{global_model_new}
\bm{\omega}_G = \frac{\sum_{m=1}^{M}  D_m \widetilde{\alpha}_m \bm{\omega}_m Q(\bm{\omega}_m)}{\sum_{m=1}^M  D_m \widetilde{\alpha}_m Q(\bm{\omega}_m)}.
\end{equation}
 
  We suppose that the cooperative prediction accuracy of the FL model at a considered iteration is $\eta$. We introduce a binary variable ${\alpha}_m$ for indicating whether the $m$th user sends a request or not and let $\widetilde{\alpha}_m$ be the inferred value of ${\alpha}_m$, so we have
\begin{equation} \label{alpha} 
\widetilde{\alpha}_m  \begin{cases}
=\alpha_m, \ \ {\rm with \ probability} \ \eta, \\ 
\neq \alpha_m, \ \ {\rm with \ probability} \ 1-\eta.
\end{cases}
\end{equation}
 
 Moreover, the binary variable $Q(\bm{\omega}_m)$ in (\ref{global_model_new}) indicates that the local model received at the edge server can be recovered correctly or not, which is given by
  \begin{equation} \label{Q} 
Q(\bm{\omega}_m)=  \begin{cases}
1, \ \ {\rm with \ probability} \ p, \\ 
0, \ \ {\rm with \ probability} \ 1-p.
\end{cases}
\end{equation}

In (\ref{Q}), $p={\rm{Pr}}\left \{\gamma_m \geq \gamma_{T} \right \}$ indicates the probability that $\bm{\omega}_m$ is received correctly at the BS, and  $1-p$ otherwise. $\gamma_m$ denotes the received SNR of $U_m$ and $\gamma_{T}$ is an SNR threshold. For instance, $Q(\bm{\omega}_m)=0$ indicates that the received local FL model parameter cannot be correctly received, otherwise, $Q(\bm{\omega}_m)=1$. 


 \subsection{An Intuitive Example}
 \textcolor{black}{
Different from the conventional FL procedure that is usually subject to the wireless bandwidth for model parameters uploading, the proposed RIS-aided FSL framework aims at improving FL performance by achieving wireless spectrum learning at the RIS controllers. Here, we consider an intuitive application scenario - the coexistence of multiple wireless systems, in which the wireless spectrum learning problem must be appropriately designed and addressed. Specifically, we suppose that two different wireless systems, such as WiFi and LTE-U systems, share the same unlicensed band and there exists interference between the two systems. To coordinate a fair share of the spectrum without causing undesired interference, improved wireless spectrum sensing and signal identification to detect spectrum users are needed~\cite{TWC_YB}. However, traditional power estimation and spectrum sensing can only detect whether the spectrum is occupied or not as in energy detection. To address this issue, the proposed FSL framework achieves appropriate wireless resources coordination of the two systems via RISs-enabled cooperative spectrum sensing, where the CNN model deployed at the RIS controllers for RF spectrum learning is trained via FL.} 

\textcolor{black}{An intuitive example with two system users (denoted by  ${\bf U} = \{U_{1}, U_{2}\}$) and two RISs is considered, where $U_{1}$ is a WiFi user and $U_{2}$ is a LTE-U user.} A CNN model is trained via the proposed FSL framework and then deployed at each RIS controller to learn the features of RF traces. 
In particular, upon receiving the superimposed incident signals at each RIS, the well-trained CNN infers the set of requested users by performing feed-forward calculation.  \textcolor{black}{Since each user has a binary state, i.e., `active' and `inactive', in this case, there exist four combinations (i.e., classes) of signals from the perspective of each RIS: 
\begin{itemize}
\item {Class-1: Idle}. This indicates that both $U_{1}$ and $U_{2}$ are inactive, so the collected RF traces include only the noise.
\item {Class-2: Only ${U_{1}}$}. This indicates that only $U_{1}$ is active, so the collected RF traces include only $U_1$.
\item {Class-3: Only ${U_{2}}$}. This indicates that only $U_{2}$ is active, so the collected RF traces include only $U_2$.
\item {Class-4: ${U_{1}\!+\!U_{2}}$}. This indicates that both $U_{1}$ and $U_{2}$ are active, so the collected RF traces include $U_1$ and $U_{2}$.
\end{itemize}}



Via inference, each RIS controller identifies the composition of the superimposed incident RF signals and sends the inferred result to the BS via a dedicated channel for data fusion. Note that the spectrum identification in the considered example boils down to a {four-class classification problem}. \textcolor{black}{According to the inferred result, the BS can coordinate spectrum resources between the two systems appropriately. Also, according to the allocated resources, the requested users upload their local trained models to the BS via RISs-aided wireless links.}

\section{System Model and Problem Formulation}  
\label{system}
 
\subsection{RIS-Aided Wireless Communication Model} 
By leveraging the presence of RISs, the users can upload the local gradients via the RIS-assisted wireless uplinks more reliably. For the local uploading of the model parameters, we consider an orthogonal-access schemes such as orthogonal frequency division multiple access (OFDMA) in an synchronous manner. 
We suppose that each RIS, denoted as $R_k$, $k \in {\cal K}$, is equipped with $N_k$ reflecting elements, which can be appropriately configured by the RIS controller to reflect the associated user's signal effectively towards the BS through non-overlapping frequency band. In general, the RISs can be appropriately deployed so that line-of-sight (LoS) links can be established between them and the BS. \textcolor{black}{Moreover, we assume that the channel state information (CSI) can be perfectly estimated at the BS by sending training symbols at the user and adjusting the reflection states of the RIS according to a predesigned training reflection pattern~\cite{CSI01,CSI02}. Then, the estimated CSI is fed back to the RIS controller via a dedicated channel\footnote{\textcolor{black}{Since channel estimation in RISs-aided wireless communications has been widely investigated, it is not explicitly considered in this paper.}}, such as \cite{jsac,CSI03}}. Note that due to the movement of the mobile users, the received signals are sometimes subject to a multipath fading and Doppler effect, which can be effectively  eliminated and/or mitigated by real-time tuneable RISs~\cite{Doppler}. 

As for $R_k$, the amplitude reflection coefficient is assumed to be equal to one for all the $N_k$ reflecting elements, so the phase reflection matrix is denoted by 
\begin{equation} \label{phi}
\mathbf{\Phi}_k={\rm{diag}} \left (e^{j \theta_1^k}, e^{j \theta_2^k},...,e^{j \theta_{N_k}^k} \right ),
\end{equation} 
where ${\bm \theta}_k=\left (\theta_1^k, \theta_2^k,..., \theta_{N_k}^k \right )$ denotes the vector of phase shifts that need to be optimized. 

In the considered FSL framework, the channels from the user $U_m$ to $R_k$ and from $R_k$ to the BS are given by $\textit{\textbf h}_{m,k} \in \mathbb{C}^{N_k \times 1}$ and $\textit{\textbf g}_{k} \in \mathbb{C}^{1 \times N_k}$, respectively. The channel gain of the direct link from $U_m$ to the BS is denoted by $h_{d,m}$. Besides, the channels are assumed to be quasi-static and remain nearly-constant during the transmission~\cite{WU01}. 
 Accordingly, the received local model parameters that are uploaded by the total $M$ users at the BS include the signals via the direct links and the links reflected by the RIS and the white Gaussian noise, i.e., 
 \begin{equation} \label{recv_sig}
{z} =\sum_{m \in {\cal M}} \left( h_{d,m}+ \sum_{k \in {\cal K}} r_{m,k} \ \textit{\textbf g}_{k} \mathbf{\Phi}_k \textit{\textbf h}_{m,k} \right) \sqrt{p_m} z_m + \bm{n},
\end{equation}
where 
 $p_m$ denotes the transmit power of $U_m$, $z_m$ is the unit-power information signals sent from $U_m$, and  $\bm{n} = [n_1, n_2, ..., n_M]^{\rm T}$ denotes the white Gaussian noise vector. 
 
 Besides, the user-RIS association matrix ${\bm R}$ with dimension of ${K \times M}$ is given by
 \begin{equation} \notag
{\bm R}=\begin{bmatrix}
r_{1,1}  &  r_{1,2}  & \cdots\ & r_{1,K} \\
r_{2,1}  &  r_{2,2}   & \cdots\ & r_{2,K}\\
 \vdots   & \vdots & \ddots  & \vdots  \\
r_{M,1} & r_{M,2}   & \cdots\ & r_{M,K}\\
\end{bmatrix},
\end{equation}
where $r_{m,k} \in \{0,1\}$, $r_{m,k}=1$ indicates that $R_k$ is allocated to $U_m$, and $r_{m,k}=0$, otherwise. We assume that only one RIS can be allocated to one user at most, so we have $\sum_{m \in {\cal M}} r_{m,k}\leq 1, \ \forall k \in {\cal K}$. 
 
\textcolor{black}{In conventional RIS-aided communication systems, it is known that each RIS reflects all the incident electromagnetic waves, and will steer them towards reflecting directions that depend on the direction of incidence and the phase shifts applied by the RIS elements. To achieve user-RIS association in the proposed FSL framework, we introduce \textbf{Assumption~\ref{assumptionO2}}.}
\vspace{-2mm}
\textcolor{black}{\begin{assumption} \label{assumptionO2}
Compared with the desired reflected signal, the interference power caused by the reflections via the remaining RISs on non-overlapping frequency bands are relatively low, which can be ignored~\cite{interferences}. Therefore, in this paper, we assume that each user will be associated with a certain RIS via a non-overlapping frequency band. 
\end{assumption}} \vspace{-2mm}

 Therefore, the SNR at the BS for user $U_m$ is given by
 \begin{equation} \label{sinr}
\gamma_m= \frac{{p_m} \left |h_{d,m}+ \sum_{k=1}^{K} r_{m,k} \textit{\textbf g}_{k} \mathbf{\Phi}_k \textit{\textbf h}_{m,k}  \right |^2 }{ B_m {\cal N}_0}, 
\end{equation}
where $B_m$ denotes the bandwidth allocated to $U_m$, and ${\cal N}_0$ indicates the noise power spectral density.  

Suppose that the total bandwidth between the users and the BS is $B$, and let ${\bm \beta}\!=\!\{\beta_1, \beta_2,..., \beta_M\}$ denote the wireless bandwidth allocation vector, so we have $B_m=\beta_m B$ and $\sum_{m=1}^{M} \beta_m \!\leq\! 1$. Based on (\ref{sinr}), the achievable data rate of $U_m$ can be obtained as
\begin{equation}\label{data_rate}
\phi_m \!=\! {\beta}_m B {\rm log}_2 \!\left(\!1\!+\! \frac{{p_m} \!\left |\!h_{d,m}\!+\! \sum_{k\!=\!1}^{K}r_{m,k} \textit{\textbf g}_{k} \mathbf{\Phi}_k \textit{\textbf h}_{m,k}  \!\right |^2\! }{{\beta}_m B {\cal N}_0} \!\right)\!.
\end{equation}

\subsection{FL Latency Model}

\subsubsection{Latency for Local Model Training}
Each user, $U_m, \forall m \in \cal M$, independently trains the local model based on the available local data samples (${\cal D}_m$)  with size $s_m$ (bits). Let $\theta \in [0, 1]$ be the accuracy for the local FL model training. The computation time, in general, depends on the number of local iterations, which is upper bounded by ${\cal O}({\rm log}(1/\theta))$ for different kinds of iterative algorithms~\cite{infocom}. In the following, we use this upper bound to approximate the number of iterations needed for the local computations by each user.  Specifically,
we introduce a positive constant $v$ that depends on the data size, and we denote the time of one local iteration of $U_m$ by $t^m_{cmp}$, so the upper bound of the computation time in one global iteration is
\begin{equation} \label{t_G}
t^m_{Gmp} = v {\rm log}(1/\theta) t^m_{cmp}.
\end{equation}

Since the computation delay ($t^m_{cmp}$) mainly depends on the user computing capability (i.e., the on-board chip CPU) and the size of data samples, for the user $U_m$, we denote by $c_m$ the number of CPU cycles to calculate the gradient with respect to one bit for each local iteration (cycles per bit), 
 and by $f_m$ the CPU frequency of $U_m$ (cycles per second). Therefore, the local computing delay of $U_m$ is calculated as
\begin{equation}\label{t_local}
t^m_{cmp} =\frac{c_m s_m}{f_m}, \ \ \forall m \in \cal M.
\end{equation}

Substituting (\ref{t_local}) into (\ref{t_G}), the computation delay  of the local FL model training in one global iteration is 
 \begin{equation} \label{t_G_new}
t^m_{Gmp} = v {\rm log}(1/\theta) \frac{c_m s_m}{f_m}.
\end{equation}

\subsubsection{Latency for Local Model Uploading}
After training the local FL model, all the users upload their model parameters to the BS via frequency domain-based multiple access, e.g., OFDMA. Given the uplink data rate in (\ref{data_rate}), the delay of uploading the local FL model parameters from $U_m$ to the BS via RISs-aided uplink is 

\begin{equation}\label{t_commu}
\begin{aligned}
t^m_{Com} \!&=\! \frac{z(\bm{\omega}_m)}{{\beta}_m B {\rm log}_2 \!\left(\!1\!+\!  \gamma_m \right)}\\
 \!&=\! \frac{z(\bm{\omega}_m)}{{\beta}_m B {\rm log}_2 \!\left(\!1\!+\! \frac{{p_m} \left |h_{d,m}\!+\! \sum_{k=1}^{K} r_{m,k} \textit{\textbf g}_{k} \mathbf{\Phi}_k \textit{\textbf h}_{m,k}  \!\right |^2\! }{{\beta}_m B {\cal N}_0} \right)}, 
\end{aligned}
\end{equation}
where $z(\bm{\omega}_m)$ indicates the data size (number of bits) of the local FL model that is sent from $U_m$.

\subsubsection{Latency for Global Model Broadcasting}
In this step, the BS aggregates the received local FL models to a global FL model. Then the BS broadcasts the global FL model parameters to the users and the RIS controllers. Since the data rate of the wireless downlinks could be relatively large due to the high transmission power of the BS and the large available bandwidth, the model broadcasting latency is neglected\footnote{This is a common assumption also made by other works such as~\cite{infocom}.}. Note that this paper considers the synchronous aggregation case, so the edge server needs to wait for the local FL model gradients from of all the users before the global aggregation can take place. This leads to the so-called \textit{straggler’s effect issue}, i.e., each training iteration only progresses as fast as the slowest user. 

Based on the above analysis, we define and compute the \textit{completion latency} for one FL global iteration.  
Consider an arbitrary communication round, the completion latency, denoted as $T$, is composed of two parts: the local training latency and the wireless uploading latency. Since all the local FL model parameters need to be uploaded to the BS to perform aggregation at each iteration, 
  the completion latency of one global iteration is calculated as
\begin{equation}\label{T}
\begin{aligned}
T &= \underset{m \in {\cal M}}{\rm max} \left \{t^m_{Gmp} +  t^m_{Com} \right \} \\
&= \underset{m \in {\cal M}}{\rm max} \left \{v {\rm log}(1/\theta) \frac{c_m s_m}{f_m} +  \frac{z(\bm{\omega}_m)}{{\beta}_m B {\rm log}_2 \!\left(\!1\!+\!  \gamma_m \right)} \right \} ,
\end{aligned}
\end{equation}
where $\gamma_m$ is given by (\ref{sinr}).

\subsection{Problem Formulation}
Having defined the latency model in the previous section, the problem is formulated to maximize the number of correctly received parameters of the local FL model per unit time at each iteration. Here, 
we define a metric for evaluating the FL performance at each iteration, called \textit{FL system utility}, which is calculated as the ratio between the total number of correctly received local FL model parameters at the BS and the total delay at the considered iteration, i.e., 
\begin{equation} \label{utility}
\begin{aligned} 
\xi & \triangleq \frac{\rm Q}{T} =\frac{\sum_{m=1}^{M} Q(\bm{\omega}_m)}{\underset{m \in {\cal M}}{\rm max} \left \{t^m_{Gmp} +  t^m_{Com} \right \}},
\end{aligned}
\end{equation}
where the binary variable  $Q(\bm{\omega}_m)=1$ indicate that the local FL model parameter is correctly received at the BS, otherwise, $Q(\bm{\omega}_m)=0$. 

In this paper, the objective is to maximize the FL system utility (${\bm  \Theta} = \{{\bm \theta}_1, {\bm \theta}_2,...,{\bm \theta}_K \}^{\rm T}$) by jointly optimizing the RIS allocation matrix (${\bm R}$), the wireless bandwidth activation vector (${\bm \beta}$), and the phase shifts matrix of the RISs whose size is ${K \times N}$\footnote{For simplicity, we assume that each RIS has the same  number of reflecting elements, i.e., $N = N_{k_1}=N_{k_2}$, where $k_1 \neq k_2$.}. 
To this end, a joint RIS phase shifts, user-RIS association and bandwidth allocation problem can be formulated as follows: 
\begin{align}
 \mathbb{P}: & \;\;\;\;\underset{\{{\bm R}, {\bm \Theta}, {\bm \beta}\}}{\rm max}\;\; \xi \label{problem}  \\ 
& \;\;\;\;\;{\rm{s}}{\rm{.t}}{\rm{.}}\;\; r_{m,k}\in \left \{0,1  \right \}, \  \forall m \in {\cal M},  \ \forall k \in {\cal K}, \tag{\ref{problem}{a}} \label{problema} \\
& \;\;\;\;\;\;\;\;\;\;\; \sum_{m =1}^{M} r_{m,k}\leq 1, \ \forall k \in {\cal K}, \tag{\ref{problem}{b}} \label{problemb}\\
& \;\;\;\;\;\;\;\;\;\;\;\;  \gamma_m^k \geq \gamma_T, \ \forall m \in {\cal M}, \tag{\ref{problem}{c}} \label{problemc}\\
& \;\;\;\;\;\;\;\;\;\;\;  \left |e^{j \theta_{n}^k} \right | = 1, \ \forall n \in [1, N_k], \ \forall k \in {\cal K}, \tag{\ref{problem}{d}} \label{problemd}  
\\
& \;\;\;\;\;\;\;\;\;\;\;  \sum_{m=1}^{M} \beta_m \leq 1. \tag{\ref{problem}{e}} \label{probleme}
 \end{align}

In (\ref{problem}a), the binary value $r_{m,k}=1$ indicates that $R_k$ is allocated to $U_m$, and $r_{m,k}=0$, otherwise.  (\ref{problem}b) indicates that at most one RIS can be allocated to a user at a time.  (\ref{problem}c) indicates that the achievable SNR  needs to be larger than a threshold. 
  (\ref{problem}d) indicates that each RIS reflecting element can only provide a phase shift $\theta_{n}^k \in [0, 2\pi)$ without amplifying the signals. Finally, (\ref{problem}e) indicates that the sum of the wireless bandwidths that are allocated to the users cannot exceed the total bandwidth.

\section{Joint Optimization of RISs Configurations and Wireless Bandwidth Allocation}\label{soving}

 We observe that the formulated problem in (\ref{problem}) is a mixed-integer nonlinear programming (MINLP), which is NP-hard and the globally optimal solution is, in general, difficult to obtain~\cite{tmc_yb}. Specifically, due to the limited number of RISs available in practice, only a subset of users may be allowed to upload, at each iteration, their local FL model parameters to the BS with the aid of RISs. Moreover, since the data samples of each user are usually non-independent and identically distributed (non-IID), the BS generally prefers to include more users' FL models to generate a converged FL model. Hence, the FL performance will be significantly affected by the user-RIS association and the wireless bandwidth allocation.                                                                                                                                                                                                                                                                                                                                                                                                                                                                                                                                                                                                                                                                                                                                                                                                                                                                                                                                                                                                                                                                                                                                                                                                                                                                                                                                                                                                                                                                                                                                                                                                                                                                                                                                                                                                                                                                                                                                                                                                                                                                                                                                                                                                                                                                                                                                                                                                                                                                                                                                                                                                                                                                                                                                                                                                                                                                                                                                                                                                                                                                                                                                                                                                                                                                                                                                                                                                                                                                                        
 
  In the following, we decompose the original problem into two subproblems, which are then solved independently. 
  \vspace{-4mm}
                                                                                                                                                                                                                                                                                                                                                                                                                                                                                                                                                                                                                                                                                                                                                                                                                                                                                                                                                                                                                                                                                                                                                                                                                                                                                                                                                                                                                                                                                                                                                                                                                                                                                                                                                                                                                                                                                                                                                                                                                                                                                                                                                                                                                                                                                                                                                                                                                                                                                                                                                                                                                                                                                                                                                                                                                                                                                                                                                                                                                                                                                                                                                                                                                                                                                                                                                                                                                                                                                                                                                                                                                                                                                                                  \subsection{Problem Decomposition and Transformation}
By exploiting the structure of the objective function and its constraints, we observe that the original problem in (\ref{problem}) has a high complexity. By using the Tammer decomposition method, we decompose the original problem into 
two subproblems with separated objective and constraints without changing the optimality of the solutions~\cite{Tammer}.  

First, we rewrite the original problem as an equivalent problem:
\begin{equation} \label{problem1}
\begin{aligned}
\widetilde{\mathbb{P}}:  & \ \underset{\bm \beta}{\rm max} \ \left(\underset{\{{\bm R}, {\bm \Theta}\}}{\rm max} \ \xi \right)   \\
&{\rm{s}}{\rm{.t}}{\rm{.}}\;\;\;(\ref{problem}a)-(\ref{problem}e).
\end{aligned}
\end{equation} 

To solve the equivalent problem $\widetilde{\mathbb{P}}$, we further decompose it into two subproblems, as illustrated in \textbf{Remark~\ref{R3}}.
\begin{remark} 
\label{R3}
Solving $\widetilde{\mathbb{P}}$ is equivalent to solving two subproblems: i) the user-RIS association subproblem while keeping fixed the bandwidth allocation vector so as to maximize the number of correctly received local FL models, i.e.,  \begin{equation} \label{RA}
\begin{aligned}
& \mathbb{P}_1: \ {\rm Q}^*=\underset{\{{\bm R}, {\bm \Theta}\}}{\rm max}\; {\rm Q}    \\
&{\rm{s}}{\rm{.t}}{\rm{.}}\;\;\;(\ref{problem}a)-(\ref{problem}d),
 \end{aligned}
\end{equation}
and, ii) the bandwidth allocation subproblem while assuming given the optimal user-RIS association to maximize the FL system utility, i.e., \begin{equation} \label{BA}
\begin{aligned}
& \mathbb{P}_2: \ \underset{\bm \beta}{\rm max}\;  \xi^*   \\
& {\rm{s}}{\rm{.t}}{\rm{.}}\;\;\;(\ref{problem}e),
 \end{aligned}
\end{equation}
where $\xi^* = {\rm Q^*}/{T}$.
 \end{remark}
 
Since ${\rm Q}^*$ in  $\mathbb{P}_2$ can be obtained by solving $\mathbb{P}_1$, $\mathbb{P}_2$ can be rewritten as the bandwidth allocation subproblem ($\widetilde{\mathbb{P}}_2$) to minimize the FL training latency:
 \begin{equation} \label{BA_new}
\begin{aligned}
& \widetilde{\mathbb{P}}_2: \ \underset{\bm \beta}{\rm min}\; {\rm max} \left \{t^m_{Gmp} +  t^m_{Com} \right \}   \\
& {\rm{s}}{\rm{.t}}{\rm{.}}\;\;\;(\ref{problem}e).
 \end{aligned}
\end{equation}

Before solving the subproblem $\mathbb{P}_1$ in (\ref{RA}),  \textbf{Observation~\ref{O1}} is presented as follows. 
\theoremstyle{Observation}
\newtheorem{observation}{\textbf{Observation}}
\begin{observation} \label{O1}
We observe from (\ref{RA}) that the optimal value of ${\bm \Theta}$ is the one that maximizes the channel gain via RISs, i.e., $\left |h_{d,m}+ \sum_{k=1}^{K} r_{m,k} \textit{\textbf g}_{k} \mathbf{\Phi}_k \textit{\textbf h}_{m,k}  \right |^2$. With this in mind, the optimal solution of $\mathbb{P}_1$ can be obtained in two steps: obtaining the optimal value of ${\bm \Theta}$ that maximizes the channel gain in the first-step, and then calculating the optimal user-RIS association matrix in the second-step. 
\end{observation}

\subsection{First-step: RISs Phase Shift Configuration}

We first optimize the phase shift vector ${\bm \Theta}$ of problem $\mathbb{P}_1$ in (\ref{RA}), where $r_{m,k}=1$. Accordingly, the optimal value of ${\bm \Theta}$ can be calculated by solving the
following problem:
\begin{equation} \label{RA_1}
\begin{aligned}
&  \underset{ {\bm \Theta}}{\rm max}\; {\left |h_{d,m}+  \textit{\textbf g}_{k} \mathbf{\Phi}_k \textit{\textbf h}_{m,k}  \right |^2}     \\
&{\rm{s}}{\rm{.t}}{\rm{.}}\;\;\; (\ref{problem}d).
 \end{aligned}
\end{equation}

According to the triangle inequality, we have 
\begin{equation} \label{RA_2}
{\left |h_{d,m}+ \textit{\textbf g}_{k} \mathbf{\Phi}_k \textit{\textbf h}_{m,k}  \right |} \leq {\left |h_{d,m} \right |}  + {\left | \textit{\textbf g}_{k} \mathbf{\Phi}_k \textit{\textbf h}_{m,k}  \right |},
\end{equation}
 where the equality holds when ${\rm arg} \left(h_{d,m} \right) = {\rm arg} \left(\textit{\textbf g}_{k} \mathbf{\Phi}_k \textit{\textbf h}_{m,k}  \right)$ is fulfilled\footnote{This suggests that the phase of the signal reflected through the user-RIS-BS links is aligned with that of the user-BS direct link.}. 
 
   We let ${\bm u}_k {\bm v}_k \triangleq \textit{\textbf g}_{k} \mathbf{\Phi}_k \textit{\textbf h}_{m,k}$, where,
\begin{equation}
 {\bm u}_k = \left [e^{j \theta_1^k}, e^{j \theta_2^k},...,e^{j \theta_{N_k}^k} \right ]  \in \mathbb{C}^{1 \times N_k}, \notag
\end{equation} 
\begin{equation}
{\bm v}_k = {\rm diag} (\textit{\textbf g}_{k}) \textit{\textbf h}_{m,k} \in \mathbb{C}^{N_k \times 1}. \notag
\end{equation}
 Then the problem in (\ref{RA_1}) is equivalent to:
\begin{align} 
& \;\;\;\;\;\;\;\; \underset{ {\bm u}_k}{\rm max}\; {\left | {\bm u}_k {\bm v}_k \right |^2}  \label{problem1}    \\
&{\rm{s}}{\rm{.t}}{\rm{.}}\;\;\; \left |{u}_k^n \right | = 1, \ \forall n \in [1, N_k], \ \forall k \in {\cal K}, \tag{\ref{problem1}{a}} \label{problem1a} \\
&\;\;\;\;\;\;\;\; {\rm arg} \left({\bm u}_k {\bm v}_k \right) = {\rm arg} \left(h_{d,m} \right), \ \forall m \in {\cal M}, \forall k \in {\cal K}. \tag{\ref{problem1}{b}} \label{problem1b}
\end{align}

We observe that the optimal solution to the problem in (\ref{problem1}) is ${\bm u}_k^*=e^{j({\rm arg} \left(h_{d,m} \right)-{\rm arg} \left({\rm diag} (\textit{\textbf g}_{k}) \textit{\textbf h}_{m,k}  \right))}$. Therefore, the $n$th phase shift of the $k$th RIS is calculated as
\begin{equation}
\theta_n^{k*} = {\rm arg} \left(h_{d,m} \right)-{\rm arg} \left(\textit{g}_{k}\right) - {\rm arg} \left( \textit{h}_{m,k} \right),
\end{equation}
where $\textit{g}^n_{k}$ and $\textit{h}_{m,k}^n$ are the $n$th element of $\textit{\textbf g}_{k}$ and $\textit{\textbf h}_{m,k}$, respectively.

As a result, we have ${\bm \theta}_k^*=\left (\theta_1^{k*}, \theta_2^{k*},..., \theta_{N_k}^{k*} \right )$. The optimal phase shifts matrix of the RISs is obtained as ${\bm  \Theta^*} = \{{\bm \theta}_1^*, {\bm \theta}_2^*,...,{\bm \theta}_K^* \}^{\rm T}$.
Having the optimal RISs phase shifts, we move to the second-step to compute the optimal user-RIS association via matching theory.
\vspace{-2mm}
\subsection{Second-step: User-RIS Association}

Via the aggregated global FL model, each RIS controller infers the set of requested users at each iteration and then sends the estimated result to the BS. Based on the fusion result, the user-RIS association can be optimized.
To tackle the combinatorial problem, we apply matching theory to map the user-RIS association subproblem into a matching game. 

Specifically, we denote the set of the requested users whose signals over the direct links cannot be correctly received as ${\cal M}'$, where ${\cal M}' \subseteq \cal M$. Denote the total number of the requested users within ${\cal M}'$ as $M'$. We solve the user-RIS association subproblem in two cases: $M' \leq K$ (case-1) and $M' > K$ (case-2). For these two cases, we first formulate the user-RIS association as a two-sided matching game, followed by the definition of the utility. Then, we present a multi-granularity based matching algorithm to achieve a stable matching.
 

\noindent{{\textit{\textbf{Case-1:} User-RIS Association via One-to-One Matching}}}

In Case-1, the number of RISs is no smaller than the number of users in ${\cal M}'$, i.e., $K \geq {M}'$. In this case, each user can be served by one RIS at a time.
\subsubsection{Matching Game Formulation}
In this multiple-RISs matching game, we consider that each user can only be associated with one RIS and each RIS can only serve one user at a time, i.e., the constraint (\ref{problem}b). Thus, this matching game is a {one-to-one matching} for user-RIS association. Since the goal of matching theory is to optimally match elements of two different sets, in our case the users set ${\cal M}$ and the RISs set ${\cal K}$, by taking into account their individual preferences, as illustrated in \textbf{Definition~\ref{D1}}.
\theoremstyle{Definition}
\newtheorem{definition}{\textbf{Definition}}
\begin{definition} \label{D1}
Let ${\cal M}$ and ${\cal K}$ be two sets of players, the matching game is given by the tuple $({\cal M}, {\cal K}, \succ_{\cal M}, \succ_{\cal K})$. Here, $\succ_{\cal M} \triangleq \{\succ_m\}_{m \in \cal M}$ and $\succ_{\cal K} \triangleq \{\succ_k\}_{k \in \cal K}$ are defined as the set of preference relations of users and RISs, respectively. 
\end{definition}

More in depth, a matching game produces a matching function $\mu$, which is defined in \textbf{Definition~\ref{D2}}.
\begin{definition} \label{D2}
A matching function $\mu$ is defined by a function from the
set ${\cal M} \cup  {\cal K}$ into the set of elements of ${\cal M} \cup  {\cal K}$ such that $m=\mu(k)$ if and only if $k=\mu(m)$.  
\end{definition}

From \textbf{Definition~\ref{D2}}, it is observed that the matching function $\mu(\cdot)$ defines a relation from a given user of the set $\cal M$ to a given RIS of the set $\cal K$ on the basis of preference relations.
It is noteworthy that the preference relations (e.g., $\succ_{\cal M}, \succ_{\cal K}$) denote the level of satisfaction of the player of one set (e.g., the user set $\cal M$) in being matched with the player of the other set (e.g., the RISs set $\cal K$) and vice versa. 
\vspace{0.5mm}
\subsubsection{Preference Lists of Users and RISs}
In the proposed game, the matching is performed by the set of users and the set of RISs using preference lists. For each player, the preference list is used to rank the players of the other set. Generally, the preferences between players belonging to the two sets are formed on the basis of the evaluation of preference functions, as defined as below.

\begin{definition} \label{D3}
Let ${\cal U}_m(k)$ and ${\cal U}_k(m)$ be the preference functions of the user $m$ and the RIS $k$, respectively. We write ${\cal U}_m(k_1) > {\cal U}_m(k_2)$ if the user $m$ prefers the RIS $k_1$ to the RIS $k_2$, and thus this situation can be given by $k_1 \succ_m k_2$. Similarly, $m_1 \succ_k m_2$ indicates that the RIS $k$ prefers the user $m_1$ to the user $m_2$, and thus ${\cal U}_k(m_1) > {\cal U}_k(m_2)$ holds.
\end{definition}
 In the following, we describe in detail the preference functions ${\cal U}_m(k)$ and ${\cal U}_k(m)$, respectively.
\begin{itemize}
\item \textit{Preference function of user $m$}: 
The preference function of user $m$,  ${\cal U}_m(k)$, is evaluated by considering the achievable channel gains of the communication links between the user $m$ and the RIS $k$ with $N_k$ elements, i.e., 
\begin{equation} \label{PF1}
{\cal U}_m(k) = \sum_{n =1}^{N_k} |h^n_{m,k}|, \ \forall m \in {\cal M}, \forall k \in {\cal K}.
\end{equation}

The intuition for this preference function comes from the objective of the users, i.e., maximization of their achievable gain via the RIS $k, \forall k \in \cal K$. Hence, based on (\ref{PF1}), the generic user $m$ ranks all the RISs in a descending order of the expected SNR, so as to construct its preference list represented by ${\cal L}_m$. Therefore, an RIS $k_1 \in \cal K$ that achieves a higher preference value (consequently the SNR achieved through the more preferred RIS is higher) based on (\ref{PF1}) will be
preferred over an RIS $k_2\in \cal K$ by the user $m$, i.e., $k_1 \succ_m k_2$. We note that channel gain between the user $m$ and the RIS $k$ (i.e., ${\bm h}_{m,k}$) can be estimated by the BS and sent back to the user. 

\item \textit{Preference function of RIS $k$}: Likewise, the generic RIS $k$ also needs to generate a preference list that ranks the users according to its preference function, i.e.,
\begin{equation} \label{PF2}
{\cal U}_k(m) \!=\! {\rm max} \left(\gamma_m^k - \gamma_T, \ 0\right), \ \forall m \in {\cal M}', \forall k \in {\cal K},
\end{equation} 
where ${\cal M}'$ is selected from the users in ${\cal M}$ with ${\gamma_m^d < \gamma_T}$, so ${\cal M}' \subseteq \cal M$ holds. $\gamma_m^d= \frac{{p_m} \left |h_{d,m} \right |^2 }{{\beta}_m B {\cal N}_0}$ is the SNR of the user $m$ via the direct link.
According to this preference function, the RIS $k$ only gives preference to the users that belong to $\cal M'$ and gives more
preference to a user $m \in \cal M'$, that results in a larger improvement of the SNR. Additionally, the users that violate (\ref{problem}d) receive a `zero' preference value and thus are ranked at the bottom of the preference list. The users whose SNR via the RIS $k$ that do not exceed the SNR threshold $\gamma_T$ will not be preferred by the RIS $k$. By doing this, the constructed preference list of the RIS $k$ represented by ${\cal L}_k$.
\end{itemize}


\subsubsection{The Proposed User-RIS Association Algorithm}
Based on these considerations, we present the user-RIS association algorithm based on the proposed one-to-one matching game, in order to find a stable matching association, which is defined as follows.
\begin{definition} \label{D4}
A matching function $\mu$ is stable if there exists no blocking
pair $(m',k')$, where $m' \in {\cal M}$, $k' \in {\cal K}$, such that $k' \succ_m \mu(m)$ and $m' \succ_k \mu(k)$, where $\mu(m)$ and $\mu(k)$ represent the current matched players of $m$ and $k$, respectively.
\end{definition}

According to \textbf{Definition~\ref{D4}}, we note that a stable solution of the proposed matching game ensures that no matched
user would benefit from modifying the assigned RIS $k$ with a new RIS $k'$. The output of the proposed user-RIS association algorithm is the RISs allocation matrix $\bm R$ that maximizes the
objective of the optimization problem $\mathbb{P}_1$. The pseudocode is given in \textbf{Algorithm~\ref{alg1}}, which is
guaranteed to converge to a stable matching via the well-known deferred acceptance algorithm. 

The proposed algorithm is performed to ensure as the final result that the users are associated with the appropriate RISs. The matching procedure begins after the completion of the initialization, including the wireless channel estimation to achieve the CSI among the users, the RISs, and the BS. The estimated CSI can be used to calculate the preference function via (\ref{PF1}) and (\ref{PF2}), respectively. So each user that is  interested in the user-RIS association can estimate the achievable SNR toward each possible RIS.

\renewcommand{\algorithmicrequire}{\textbf{Initialization:}} 
\renewcommand{\algorithmicensure}{\textbf{Output:}} 
\begin{algorithm}[t]    
\caption{User-RIS Association Algorithm for One-to-one Matching}             
\footnotesize
\label{alg1}                  
\begin{algorithmic}[1]             
\REQUIRE $p_m$, $h_{d,m}$, ${\bm h}_{m,k}$, ${\bm g}_k$, $\mathbf{\Phi}_k$, $\gamma_T$, ${\cal N}_0$,  
 $t=1$, $\mu^{(1)}\triangleq\{\mu(k)^{(1)}, \mu(m)^{(1)}\}_{k\in {\cal K}, m\in {\cal M}}=\O$;
\STATE  \textbf{for} user $m\in \cal M$, constructs its preference list on all RISs according to (\ref{PF1}), denoted as ${\cal L}_m$;
\STATE  \textbf{end for}
\STATE  \textbf{for} RIS $k\in \cal K$, constructs its preference list on all users according to (\ref{PF2}), denoted as ${\cal L}_k$;
\STATE  \textbf{end for}
\STATE \textbf{repeat:}
\STATE \ \ $t\leftarrow t+1$;
\STATE \ \ \textbf{for} $m\in \cal M$, proposes $k$ according to ${\cal L}_m$ \textbf{do}
\STATE \ \ \ \ \  \textbf{if} {$\gamma_m^k > \gamma_T$} \textbf{then}
\STATE \ \ \ \ \ \ \ \  RIS $k$ checks its preference list ${\cal L}_k$;
\STATE \ \ \ \ \ \ \ \ \textbf{if} $m \succ_k \mu(k)^{(t)}$ \textbf{then}
\STATE \ \ \ \ \ \ \ \ \ \ \ $\mu(k)^{(t)} \leftarrow m$;
\STATE \ \ \ \ \ \ \ \ \textbf{else} $m$ is rejected;
\STATE \ \ \ \ \  \textbf{else} $m$ is rejected;
\STATE \ \ \textbf{end for}
\STATE  \textbf{until}  $\mu^{(t)} = \mu^{(t-1)}$.
\ENSURE $\mu^{(t)}$.\\  
\end{algorithmic}
\end{algorithm}

During the one-to-one matching, each unassigned user $m$ and RIS $k$ constructs their preference lists according to (\ref{PF1}) and (\ref{PF2}), respectively (\textit{lines 1-4} in \textbf{Algorithm~\ref{alg1}}). As illustrated in \textit{lines 7-14} of \textbf{Algorithm~\ref{alg1}}, the user $m$ proposes its most preferred RIS according to ${\cal L}_m$. If (\ref{problem}d) is violated, the user $m$ is rejected. Otherwise, the RIS $k$ checks its preference list ${\cal L}_k$. If ranked higher than the current match, i.e., $m \succ_k \mu(k)$, the user $m$ will be accepted. Otherwise, it will be rejected.  This one-to-one matching process is carried out iteratively until a stable matching function $\mu$ is found between both sets of users and RISs. The matching algorithm will converge when the matching of two consecutive iterations remains unchanged, i.e., $\mu^{(t)} = \mu^{(t-1)}$.

\vspace{1mm}
\noindent \textit{{{\textbf{Case-2:} User-RIS Association via One-to-Many Matching}}}

In Case-2, the number of RISs is smaller than the number of users in ${M}'$, i.e., $K < {\cal M}'$, so the elements of each RIS can be divided into multiple groups~\cite{group}. In this case, multiple users are allowed to reuse one RIS in such a manner that the association is no longer constrained by (\ref{problem}b), while does not violate the SNR constraint in (\ref{problem}c)\footnote{Since the elements within one group are associated with a certain user with an optimized phase shift and reflection amplitude, the impact of signal reflection via other element group mainly depends on the angle of incidence~\cite{YB_TVT}, which is ignored in this paper for simplicity.}.  
 Thus, this matching game is considered as a {one-to-many matching} with externalities for the user-RIS association.

The preference function of the user $m$, i.e., ${\cal U}_m(k)$, is defined as
 \begin{equation} \label{PF_case2}
{\cal U}_m(k) = \sum_{n =1}^{N_k^R} |h^n_{m,k}|, \ \forall m \in {\cal M}, \forall k \in {\cal K},
\end{equation}
where $N_k^R=N_k-N_k^O$ indicates the number of unoccupied elements of the RIS $k$, and $N_k^O$ indicates the number of elements of the RIS $k$ that have been occupied already. 

We denote the number of users associated with the RIS $k$ as $M_k$, and we let the total number of elements of the RIS $k$ occupied by the associated user $m$ be $n_{m,k}$, which can be calculated by solving the equation $\gamma_m^k = \gamma_T$. So the number of elements of the RIS $k$ that are already occupied by the users is
$N_k^O = \sum_{m=1}^{M_k} n_{m,k}, \  \forall k \in {\cal K}$.
From $N_k^O$,  $N_k^R$ can be obtained accordingly.
 
The preference function of the RIS $k$, i.e., ${\cal U}_k(m)$, is defined as  
\begin{equation} \label{PF3}
{\cal U}_k(m) \!=\! \underset{i}{\rm max}\left \{|{\cal M}_i'|: \gamma_{{\cal M}_i'} \geq \gamma_T  \right \}, \ \forall k \in {\cal K}.
\end{equation}

\begin{algorithm}[t]    
\caption{User-RIS Association Algorithm for One-to-many Matching}             
\footnotesize
\label{alg2}                  
\begin{algorithmic}[1]             
\REQUIRE $p_m$, $h_{d,m}$, ${\bm h}_{m,k}$, ${\bm g}_k$, $\mathbf{\Phi}_k$, $n_{m,k}$, $\gamma_T$, ${\cal N}_0$, 
 $t=1$, $\mu^{(1)}\triangleq\{\mu(k)^{(1)}, \mu(m)^{(1)}\}_{k\in {\cal K}, m\in {\cal M}}=\O$;
\STATE \textbf{for} user $m\in \cal M$, constructs its preference list on all RISs according to (\ref{PF_case2}), denoted as ${\cal L}_m$; 
\STATE  \textbf{end for}
\STATE  \textbf{for} RIS $k\in \cal K$, constructs its preference list on all users according to (\ref{PF3}), denoted as ${\cal L}_k$;
\STATE  \textbf{end for}
\STATE \textbf{repeat:}
\STATE \ \ $t\leftarrow t+1$;
\STATE \ \ \textbf{for} $m\in \cal M$, proposes $k$ according to ${\cal L}_m$ \textbf{do}
\STATE \ \ \ \ \  \textbf{if} {$\gamma_m^k > \gamma_T$ } \textbf{then}
\STATE \ \ \ \ \ \ \ \  RIS $k$ checks the number of unoccupied elements, $N_k^R$;
\STATE \ \ \ \ \ \ \ \  \textbf{if} {$N_k^R \geq  n_{m,k}$ } \textbf{then}
\STATE \ \ \ \ \ \ \ \ \ \ \ RIS $k$ accepts the proposal $m$;
\STATE \ \ \ \ \ \ \ \  \textbf{else} {RIS $k$ checks its preference list ${\cal L}_k$;}
\STATE \ \ \ \ \ \ \ \ \ \ \  \textbf{if} $m \succ_k \mu(k)^{(t)}$ \textbf{then}
\STATE \ \ \ \ \ \ \ \ \ \ \ \ \ \ $\mu(k)^{(t)} \leftarrow m$;
\STATE \ \ \ \ \ \ \ \ \ \ \ \textbf{else} $m$ is rejected;
\STATE \ \ \ \ \  \textbf{else} $m$ is rejected;
\STATE \ \ \textbf{end for}
\STATE \ \ $\cal M \leftarrow \cal M_r$, where ${\cal M_r}$ is the set including the users that are not associated with RIS yet;
\STATE \ \  \textbf{for} user $m\in \cal M$, updates its preference list on all RISs according to (\ref{PF_case2}); 
\STATE  \ \ \textbf{end for}
\STATE  \textbf{until}  $\mu^{(t)} = \mu^{(t-1)}$.
\ENSURE $\mu^{(t)}$.\\  
\end{algorithmic}
\end{algorithm}

According to (\ref{PF3}), each RIS chooses a subset of users ${\cal M}_i'$, ${\cal M}_i'  \subseteq {\cal M}'$, such that the SNR of each user in the subset ${\cal M}_i'$ is no smaller than a tolerable SNR threshold $\gamma_T$. The preference function defined in (\ref{PF3}) aims to maximize the number of elements included in the subset ${\cal M}_i'$, i.e., $|{\cal M}_i'|$. This allows the users that achieve the highest SNR to be preferred by the RIS $k$. Therefore, the subset with the largest number of elements is the most preferred among all the feasible subsets and is ranked accordingly.
The user-RIS association algorithm to find a stable matching association for the proposed one-to-many matching game is illustrated in \textbf{Algorithm~\ref{alg2}}.
Different from the previous one-to-one matching, in the proposed one-to-many matching problem, we have to take into account that the preference list of each user is dependent on the others users' preferences, as highlighted in \textbf{Remark~\ref{R2}}.

\begin{remark} \label{R2}
Due to the externalities in the proposed one-to-many matching problem, the user's preference of choosing an RIS in (\ref{PF_case2}) is mutually influenced by the number of RIS elements already utilized by other users. As a consequence, the preferences list of each user should be updated upon each association, as shown in lines 18-20 in \textbf{Algorithm~\ref{alg2}}.
\end{remark}

 

Then, we have obtained the phase shifts matrix of the RISs (denoted as ${\bm  \Theta^*}$) and the user-RIS association matrix (denoted as ${\bm R}^*$), from which ${\rm Q^*}$ can be calculated accordingly. Next, we solve $\widetilde{\mathbb{P}}_2$ to minimize the FL latency.

 \vspace{-4mm}
\subsection{Bandwidth Allocation for FL Latency Minimization}
We solve the subproblem $\widetilde{\mathbb{P}}_2$ to minimize the FL latency, as described in {\textbf{Theorem~\ref{Tm1}}}.

\begin{theorem}
\label{Tm1}
The solution to $\widetilde{\mathbb{P}}_2$ is as follows:
\begin{equation}\label{solution}
\begin{aligned}
{\beta}_m^* & = {\rm arg \underset{\bm \beta}{\rm min}} \ T \\
& = \frac{z(\bm{\omega}_m)}{ B {\rm log}_2 \left(1\!+\! \gamma_m \right) \left(T^*-v {\rm log}(1/\theta) \frac{c_m s_m}{f_m}\right)},  
\end{aligned}
\end{equation}
where $\gamma_m$ is defined in (\ref{sinr}), and $T^*$ denotes the minimal completion latency that satisfies the following conditions:
\begin{equation}\label{condition1}
\sum_{m=1}^{M} \frac{z(\bm{\omega}_m)}{ B {\rm log}_2 \left(1\!+\! \gamma_m \right) \left(T^*-v {\rm log}(1/\theta) \frac{c_m s_m}{f_m}\right)} \leq  1.
\end{equation}
\end{theorem} 

\begin{proof}

In the proposed RIS-assisted FL structure, if the training and uploading of the local FL model of some users is slower than others, the BS will allocate more wireless bandwidth to the slower users to accelerate the FL procedure. As a result, the FL system completion latency can be shortened. 
Therefore, the optimal bandwidth allocation vector $\bm \beta$ can be achieved only when all the users finish the training of the local FL model and the wireless uploading at the same time (denoted as $T^*$), i.e., \begin{equation}\label{T*}
\begin{aligned}
T &= \underset{m \in {\cal M}}{\rm max} \left \{v {\rm log}(1/\theta) \frac{c_m s_m}{f_m} +  \frac{z(\bm{\omega}_m)}{{\beta}_m B {\rm log}_2 \!\left(\!1\!+\!  \gamma_m \right)} \right \} \\
& \triangleq T^*.
\end{aligned}
\end{equation}

Substituting $\beta_m = \beta_m^*$ into (\ref{T*}), we have 
\begin{equation} \label{problem_new}
v {\rm log}(1/\theta) \frac{c_m s_m}{f_m} +  \frac{z(\bm{\omega}_m)}{{\beta}_m^* B {\rm log}_2 \!\left(\!1\!+\!  \gamma_m \right)} = T^*.  
\end{equation}

By solving (\ref{problem_new}), $\beta_m^*$ can be derived as (\ref{solution}). Note that $T^*$ fulfills (\ref{condition1}) and can be obtained using the bisection search, e.g., through interval halving or binary search~\cite{binary_search}.
As a result, $\beta_m^*$ can be calculated accordingly.
\end{proof}

In the proposed FSL framework, the deployment of RISs has a twofold contribution: 1) the RIS controllers help detecting the set of requested users in the resource allocation stage,  and 2) the RISs can improve the data transmissions in the FL training stage. Based on this, the performance improvement given by deployment of RISs is presented in \textbf{Theorem~\ref{Tm0}}.

\begin{theorem}
\label{Tm0}
For a considered iteration in the proposed FSL framework, let $\epsilon$ be the prediction accuracy of the FL model, and let $K$ be the number of RISs. The performance improvement brought by RISs in the first stage is given by $\chi_1=\frac{1-(1-\epsilon)^K}{\epsilon}$, and the improvement brought by RISs in the second stage is 
$\chi_2=e^{\kappa}$, where 
\begin{equation} \nonumber
 \begin{aligned}
\kappa \triangleq\frac{{{p_m}} }{{\beta}_m B {\cal N}_0} \left |\sum_{k=1}^{K} r_{m,k} \textit{\textbf g}_{k} \mathbf{\Phi}_k \textit{\textbf h}_{m,k}  \right |^2 \\
\;\;\; +   \frac{2\sqrt{\gamma_m^d p_m} }{\sqrt{{\beta}_m B {\cal N}_0}} \left |\sum_{k=1}^{K} r_{m,k} \textit{\textbf g}_{k} \mathbf{\Phi}_k \textit{\textbf h}_{m,k}  \right |,
  \end{aligned}
\end{equation}
and $\gamma_m^{d}=\frac{{p_m} \left |h_{d,m} \right |^2 }{{\beta}_m B {\cal N}_0}$.
\end{theorem} 

\begin{proof}
Please refer to Appendix~A.
\end{proof}

From \textbf{Theorem~\ref{Tm0}}, we observe that the  performance improvement brought by RISs in the proposed FSL framework  is proportional to the number of RISs.\vspace{-5mm}

\textcolor{black}{
\subsection{Feasibility Analysis} 
In this section, we demonstrate the feasibility of the proposed scheme by clarifying the RIS elements design, analyzing the signaling overhead, storage, as well as quantifying the computation complexity at the RIS controller.
\subsubsection{RIS elements design}
In our proposed RIS structure design, each RIS consists of two types of elements: the conventional reflecting elements for signal reflection, and the {\textit{semi-active elements}} that are used for incident radio frequency (RF) signal processing. To be specific, for the semi-active elements in Fig.~\ref{f01}, only RF front-end, ADC, and down-conversion are required for obtaining I/Q sequences, and baseband processing such as signal decoding is not necessary. Different from the full-active RIS elements design for sensing and reflection~\cite{Active_sensing, Active_sensing_add, Active_sensing_2}, our proposed semi-active RIS elements involve lower power consumption and lower hardware complexity. 
\begin{figure}[t] 
  \captionsetup{font={footnotesize }}
\centerline{ \includegraphics[width=3.3in, height=2.35in]{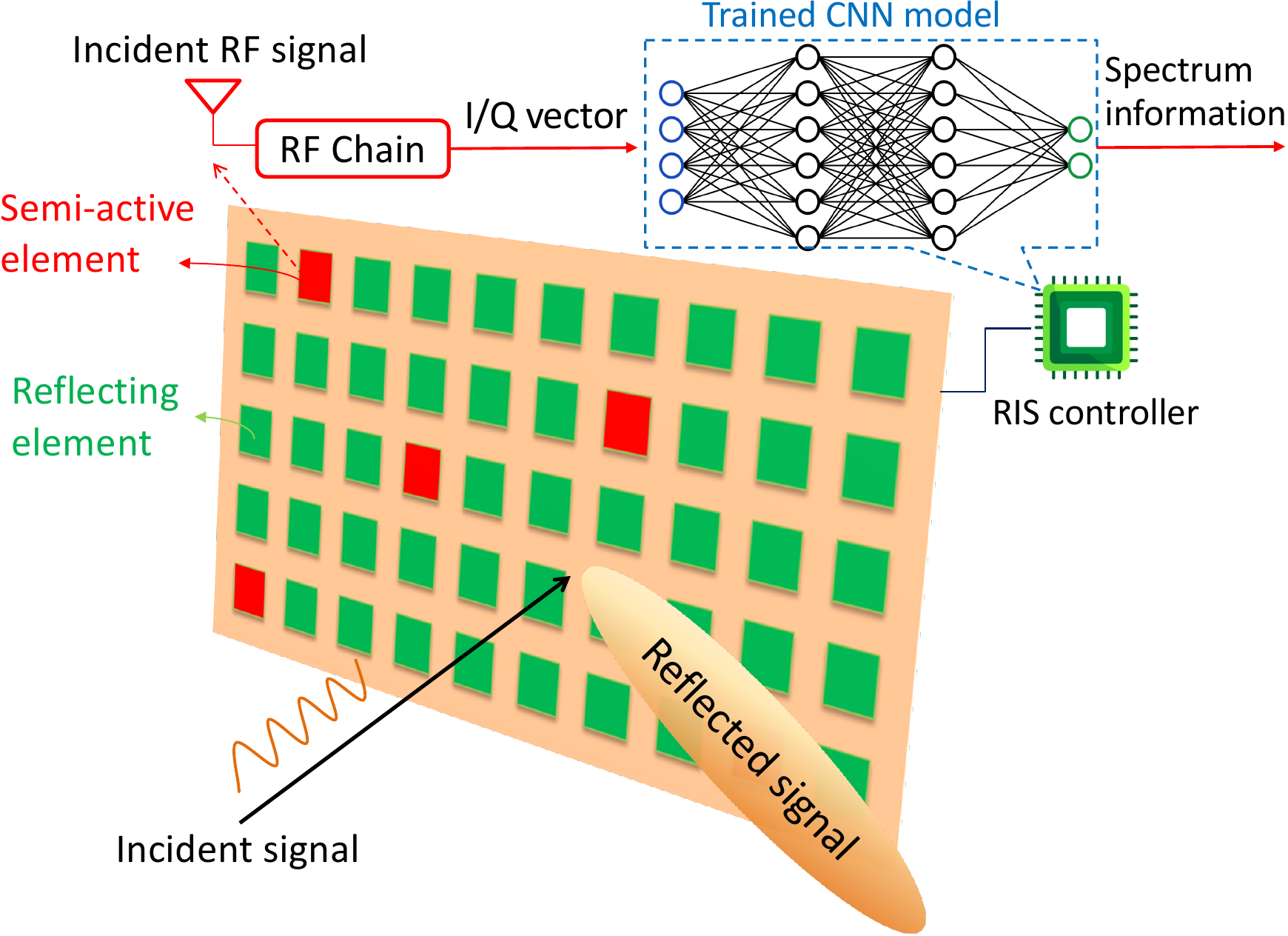}}
\caption{\textcolor{black}{RIS-aided spectrum sensing structure with semi-active RIS elements.}}
\label{f01}
\end{figure}}

\textcolor{black}{The existence of semi-active elements may lead to a loss of signal reflection\cite{RIS_sensing_add}, and thus the ratio/proportion of the semi-active elements need to be carefully chosen. In our case, deterioration of the signal reflection is not an issue since only a very small fraction of semi-active elements is required for obtaining the I/Q sequences without need for the baseband processing of the received RF signal\footnote{\textcolor{black}{According to~\cite{Active_sensing}, the achievable rate can be maximized with only a small fraction of active elements (1$\%$ and 7$\%$ for a high-frequency 28 GHz scenario and for a low-frequency 3.5 GHz scenario). The received signal can also be well recovered by deploying about 10$\%$ of the sensing elements~\cite{Active_sensing_2}.}}.
\subsubsection{Signaling overhead}
Compared to the conventional FL procedure,  the additional signaling overhead introduced in the proposed FSL framework includes the overhead for transmitting the CNN inference result and the updated model parameters via a dedicated channel. Due to the relatively small size of inference result and the high transmission power of the BS, the signaling overhead is very low and usually can be neglected.
\subsubsection{Storage analysis} 
Due to the diversity of the input RF I/Q profile, the CNN model needs to be trained for a wide range of SNR to adapt to the dynamic wireless channel. After the FL training, the fully-trained CNN model is stored at each RIS controller for online inference. Since the size of the pre-trained CNN model is relatively small\footnote{\textcolor{black}{Because of the limited computational resources of the RIS controller, the well pre-trained CNN model, e.g., around 2 Mbytes with about 300 thousand parameters~\cite{dami}, can be further quantized and become much smaller, e.g., within the Kbytes range.}}, the CNN model can be even directly cached into the memory of the RIS controller in advanced in order to perform the inference more efficiently.
\subsubsection{Computational complexity}
 To implement the proposed FSL framework, the major computational complexity lies in solving the phase shift optimization problem, the wireless bandwidth allocation problem, user-RIS association, and spectrum sensing via CNN inference. 
\begin{itemize}
\item Complexity for solving the phase shift optimization problem at the BS: To calculate the optimal phase shift, the complexity lies in computing ${\bm u}_k^*=e^{j({\rm arg} \left(h_{d,m} \right)-{\rm arg} \left({\rm diag} (\textit{\textbf g}_{k}) \textit{\textbf h}_{m,k}  \right))}$ for each RIS element. Denote the total number of elements of all the RISs as $J=\sum_{k=1}^{K}N_k$, where $N_k$ indicates the number of elements for the $k$th RIS. Then, the resulting complexity is ${\cal O}(J)$.
\item Complexity for solving the wireless bandwidth allocation problem at the BS: By using the bisection search, the complexity lies in checking the feasibility condition (\ref{condition1}). Therefore, the computational complexity is ${\cal O}(M {\rm log}_2(1/\epsilon))$ with accuracy $\epsilon$.
\item Complexity for user-RIS association at the RIS controllers: Assume worst case when the preferences of all the users for all the RISs are the same, the complexity is linear in the size of the input preference profiles~\cite{matching}, i.e., ${\cal O}(MK)$ with $M$ users and $K$ RISs.
\item Complexity for spectrum sensing via CNN inference at the RIS controllers: The trained CNN model has a quadratic time complexity during the inference process, i.e., ${\cal O}(n^2 c)$, where $c$ denotes the number of layers and $n$ denotes the number of neurons at each layer. Then, the resulting complexity for spectrum sensing is  ${\cal O}(n^2 c K)$ with $K$ RISs.
\end{itemize}}


\textcolor{black}{As a result, only a quadratic computational complexity, ${\cal O}(M+n^2 c)$, is obtained at each RIS controller. 
Since there is no exponential computational complexity, large communication overhead, or megabyte of data storage, our proposed RIS-aided FSL framework appears to be feasible in practical implementations and deployments.}
\section{Simulation Resutls and Discussions}\label{result}

 \subsection{Testing Results}
In this section, we present the testing results of the considered four-class inference example, by using real RF samples. In the considered example with two users and two RISs, the CNN model is trained through the proposed FSL framework to achieve RF signal classification. The testing results are illustrated in Table~\ref{accuracy_good}.

 \begin{table}[t]
  \captionsetup{font={footnotesize}} 
\caption{Inference accuracy of the converged CNN model.}   
 \label{accuracy_good}
\centering

\begin{tabular}{|c | l | l | l| } 
\hline
\textbf{Inferring Class} & $w=32$ & $w=128$ & $w=512$\\
\hline
Class-1 & 99.96$\%$ & 100.00$\%$ & 99.98$\%$ \\
\hline 
Class-2  &  98.51$\%$ & 98.10$\%$ & 96.23$\%$\\
\hline
Class-3   & 96.12$\%$ & 96.24$\%$ & 95.58$\%$ \\
\hline
Class-4 & 99.04$\%$ & 99.62$\%$ & 99.78$\%$  \\
\hline
\end{tabular}
\end{table}

 For each user in the considered scenario, historical RF traces are collected using a universal software radio peripheral (USRP2) testbed, which is wired connected via Gigabit Ethernet to a host server
with an implementation of the GNU Radio. To be specific, each user is emulated through a laptop that is responsible for baseband processing while a USRP2 platform is used for the up/down-conversion, the digital-to-analog/analog-to-digital conversion, and wireless transmission of the
signals. \textcolor{black}{Considering that each USRP2 testbed usually has one transmuting antenna, in order to collect the signals with USRP2, we let multiple USRP2 testbeds transmit RF signals simultaneously to a receiver.} Then the RF signals can be received and stored as I/Q sequences for training the CNN model, by including the wireless channel, for a wide range of SNR values (from $0$ to $20$ dB with interval of $5$ dB) in order to account for different channel conditions. Besides, the window size $w$ (i.e., the number of time steps of the collected RF data) in Table~\ref{accuracy_good} is $32$, $128$, and $512$, respectively.
 
 In the considered example, we assume that the RIS$_1$ and RIS$_2$ assist the wireless transmissions of $U_1$ and $U_2$, respectively, and the total wireless bandwidth is shared equally between the two users. For the local training, the users train the CNN with $80\%$ of RF dataset (i.e., I and Q samples), validate it by using $10\%$ of the dataset, and test it by using the remaining $10\%$ of the dataset. The trained CNN model consists of two convolutional (Conv) layers with ReLU activation functions, followed by two dense fully connected (FC) layers. The trained CNN model includes $256$ filters (1$\times$3) in the first Conv layer, $128$ filters (1$\times$3) in the second Conv layer, $256$ neurons in the first FC layer, and $9$ neurons in the second FC layer (output).
Within each FL round, the two users upload their local models once the local training ends. It is observed from Table~\ref{accuracy_good} that the inference accuracy of the converged CNN model with two RISs is, in general, greater than $95\%$. \textcolor{black}{Compared to other classes, the `Idle' class has the main characteristic that no user transmits and only background noise exists. Due to the {distinguishable pattern compared to the other three classes}, the CNN model predicts the `Idle' class (i.e., the Class-1 in Table~\ref{accuracy_good}) perfectly.}

\subsection{Simulation Setting}

In the considered simulation model, we consider $M$ users, $K$ RISs and one BS co-located with an edge server. The users are uniformly distributed in a square area of size $50\times50$ (in meters) with the BS is located at ($0$, $0$, $100$) in a three-dimensional Cartesian coordinates system. The location of the $k$th RIS is given by ($x_k$, $y_k$, $z_k$) = ($\frac{100}{k} {\rm cos} (45^o)$, $\frac{100}{k} {\rm cos} (45^o)$, $50$), as illustrated in Fig.~\ref{simulation_set}.
\begin{figure}[t]
  \centering
  \captionsetup{font={footnotesize }}
  \includegraphics[width=1.95in, height=1.75in]{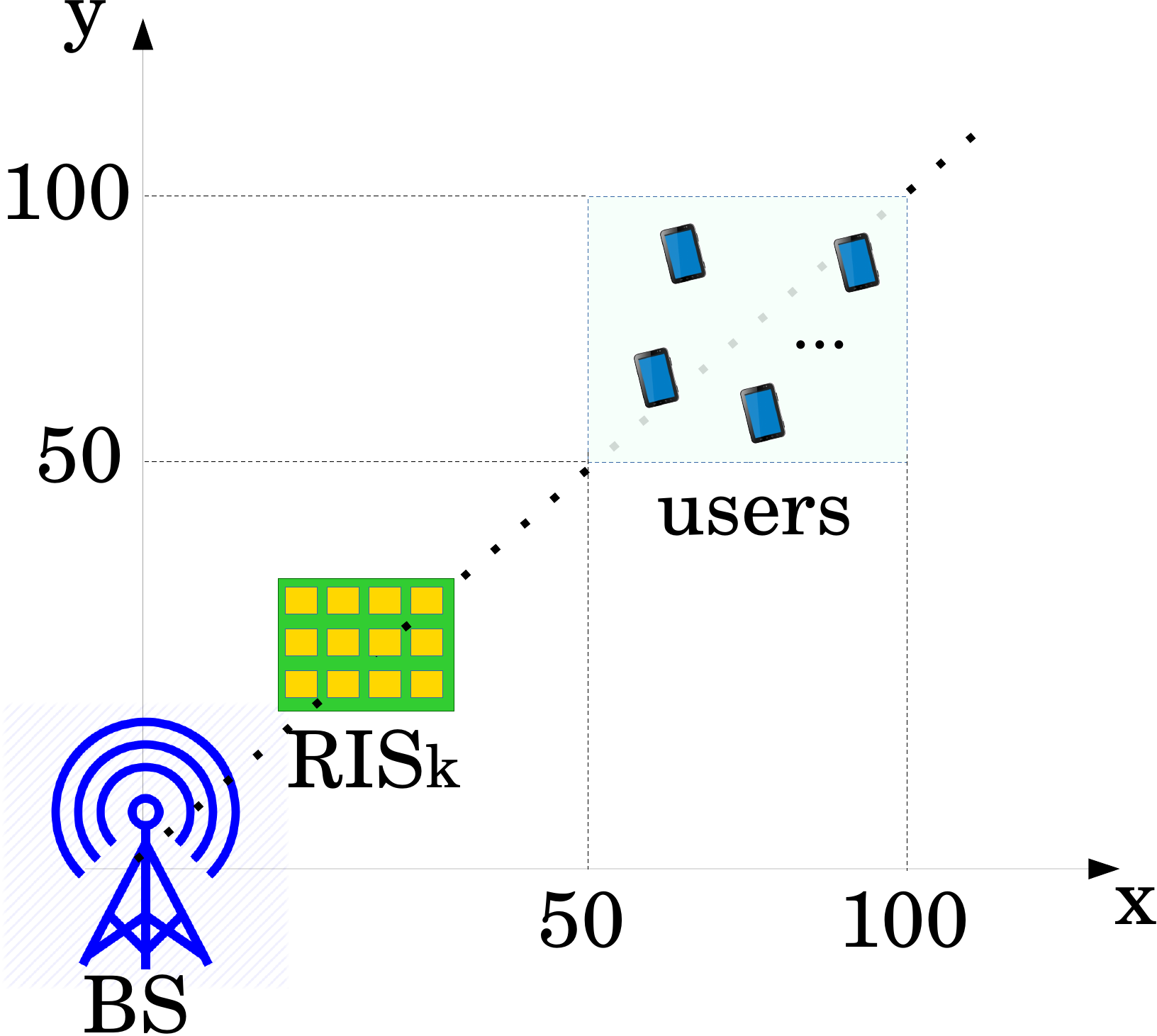}
  \caption{Simulation setup of the proposed FSL framework with multiple RISs (top view).}
  \label{simulation_set}
\end{figure}

 Unless stated otherwise, other simulation parameters are set as follows. Each RIS is equipped with $128$ reflecting elements. The transmit power of each user is $p_m = 100$ mW; the bandwidth is $B = 1$ MHz and the corresponding noise power density is ${\cal N}_0=-104$ dBm/Hz. The wireless channel gains are modeled using the 3GPP Urban Micro  with a carrier frequency of $3$ GHz. The local FL accuracy is $\theta=0.1$ and $v=1$. The computation ability of each user is $f_m=2.0$ GHz and the data size is $200$ kbits. The proposed FSL framework is simulated by using the Matlab Deep Learning Toolbox for RF fingerprinting.
 For comparison, we consider a {benchmark}: an FL algorithm that randomly determines the user-RIS association and the wireless bandwidth is equally allocated to all users.
 All statistical results are averaged over $10^4$ independent runs.

\vspace{-3mm}
\subsection{Prediction Accuracy}
\subsubsection{Prediction accuracy and training loss versus the number of iterations}
The prediction accuracy and training loss of the proposed FSL framework versus the number of iterations are shown in Fig.~\ref{iterations}(a) and Fig.~\ref{iterations}(b), respectively, where the number of RISs is $1$, $2$, and $4$, and the number of users is 2. From Fig.~\ref{iterations}(a), it is observed that as the number of iterations increases, the prediction accuracy of all considered schemes increases first and, then remains stable. The figures show that the FL algorithm converges after more than $200$ iterations. Accordingly, the training loss of all schemes shown in Fig.~\ref{iterations}(b) decreases as the number of iterations increases. Also we observe from Figs.~\ref{iterations}(a)-(b) that, as the number of RISs increases, the proposed FSL framework with more RISs can achieve a higher prediction accuracy and lower training loss.  This is because, as the number of RISs increases, cooperative spectrum sensing has a more pronounced effect and the number of local model parameters used for FL aggregation increases, thereby achieving better FL training.

\begin{figure}[t] 
\centering
\captionsetup{font={footnotesize }}
\subfigure[]{
\includegraphics[width=2.8in,height=1.55in]{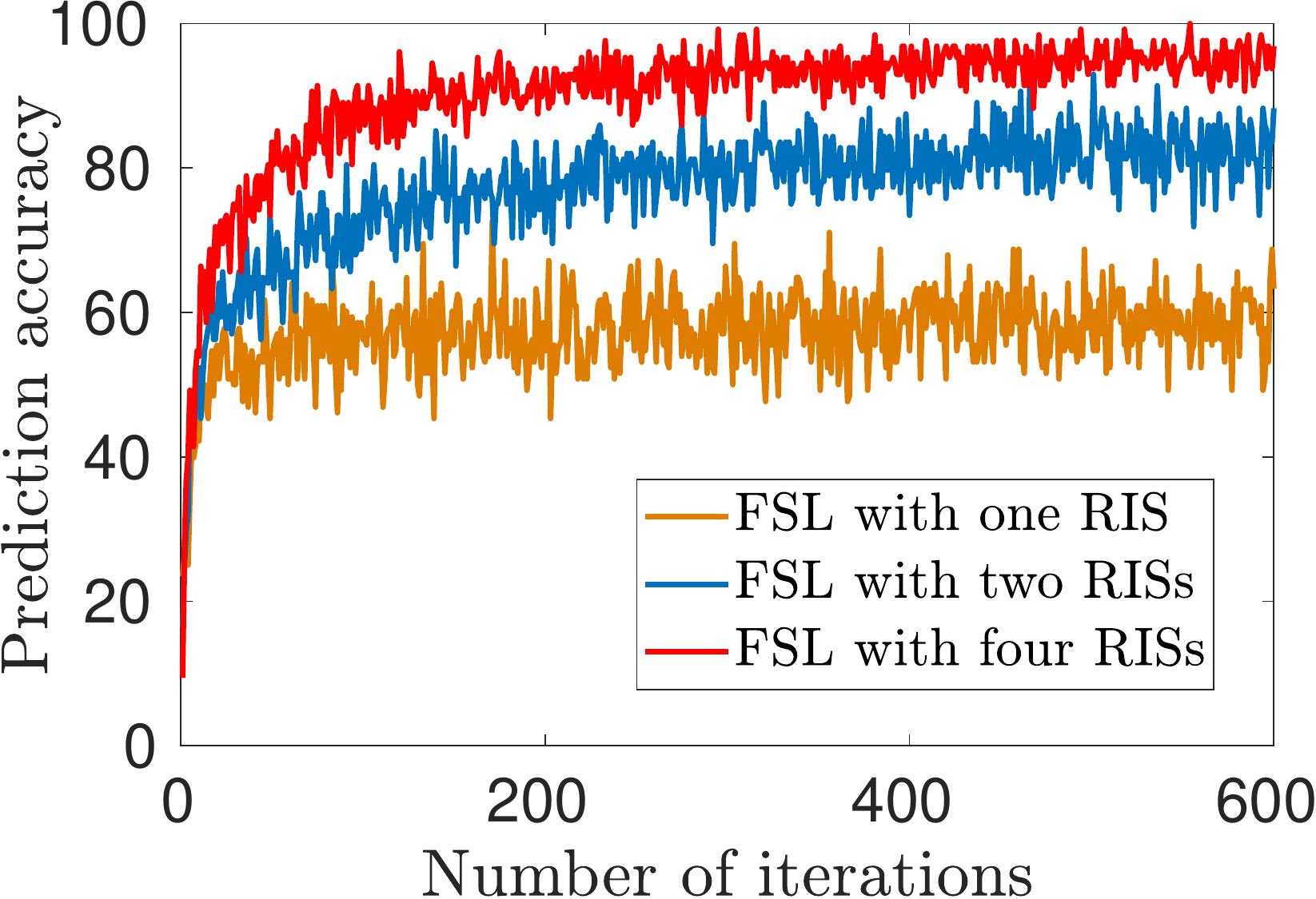}}
\hspace{3mm}
\subfigure[]{
\includegraphics[width=2.8in,height=1.55in]{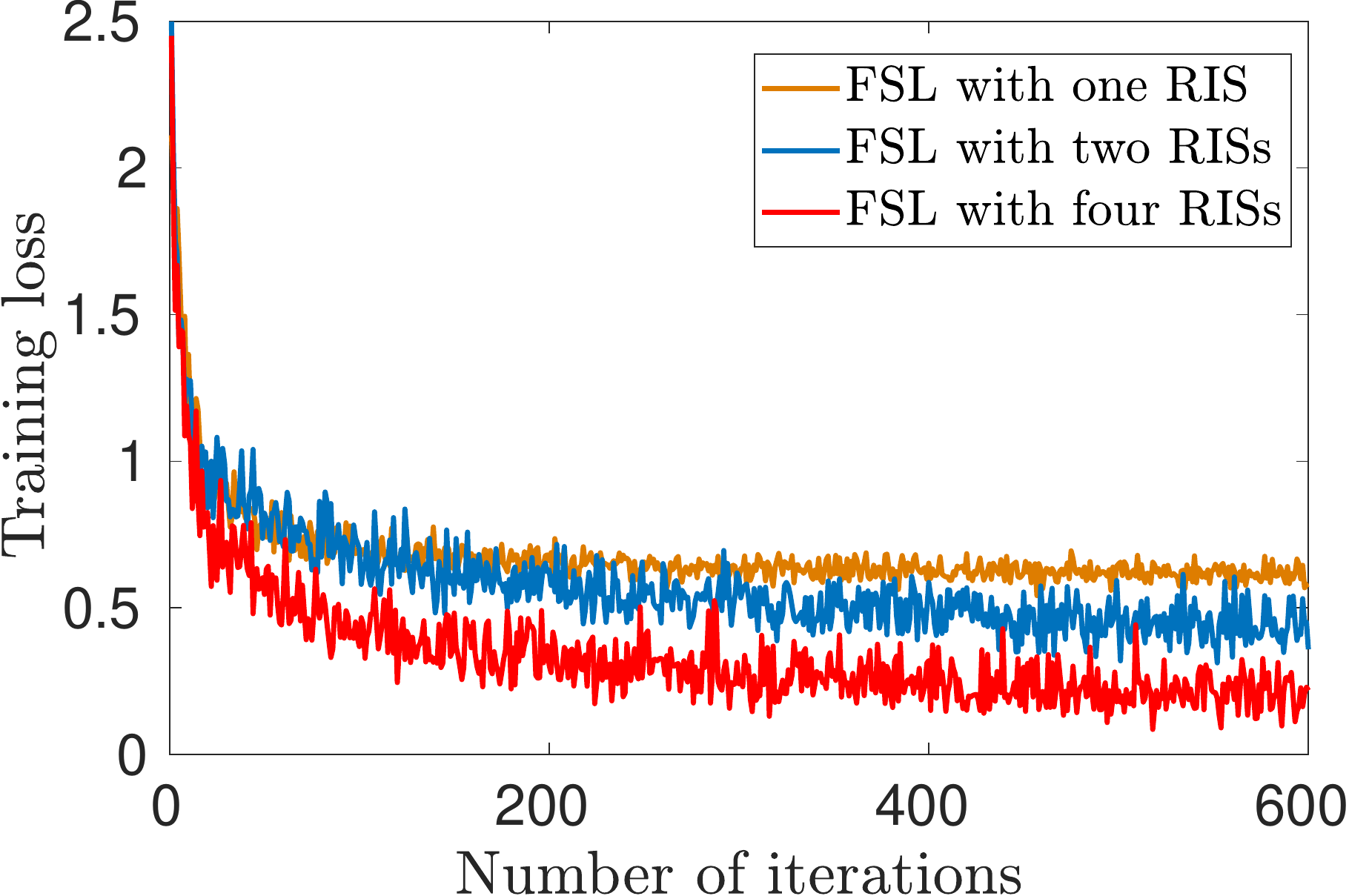}}
\caption{The prediction accuracy of the proposed FSL versus the number of iterations  is shown in (a), the training loss of the proposed FSL versus the number of iterations  is shown in (b), where the number of RISs are $1$, $2$, and $4$, respectively. }
\label{iterations} 
\end{figure}

\begin{figure}[t]
  \centering
  \captionsetup{font={footnotesize }}
  \includegraphics[width=3.1in, height=2.55in]{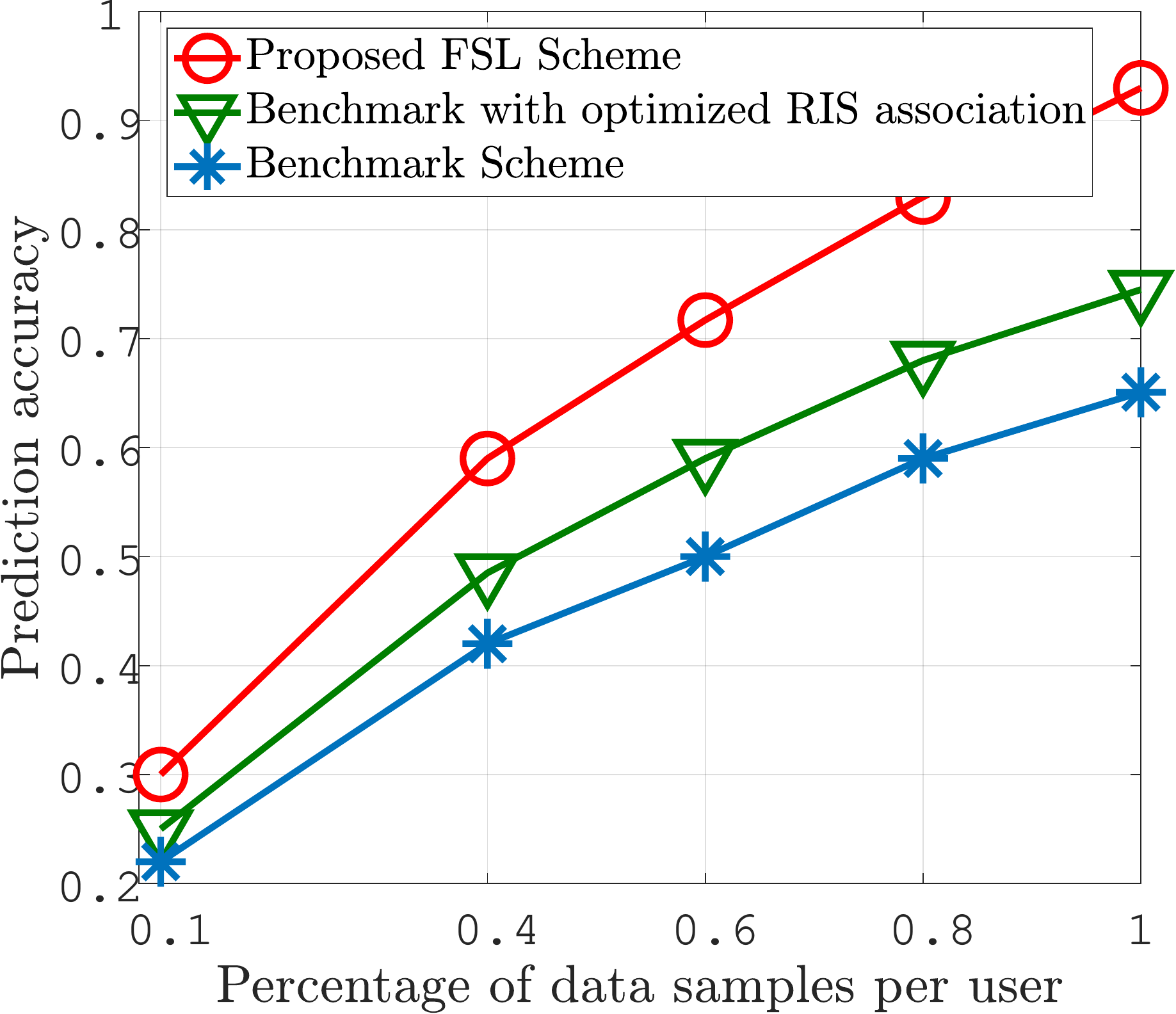}
  \caption{The prediction accuracy vs. the percentage of data samples of each user, where the number of RISs is $4$.}
  \label{accuracy_samples}
\end{figure}
\subsubsection{Prediction accuracy versus the number of training samples}
The prediction accuracy of three considered schemes as the number of training samples varies is illustrated in Fig.~\ref{accuracy_samples}, where there exist two users and four RISs. From Fig.~\ref{accuracy_samples}, we can observe that, as the percentage of training samples increases, the prediction accuracy of three schemes increases accordingly. This is because, as the number of training data samples increases, each user can use more data
samples to train their local FL models, thereby improving the prediction accuracy of the FL. We also see that when each user has $100 \%$ percentage of data samples for local training, the proposed FSL scheme can improve the prediction accuracy by up to $25\%$ and $43\%$, respectively, compared to the benchmark and that with an optimized association. These gains stem from the fact that a matching scheme is developed for user-RIS association to maximize the number of users served by RISs. Meanwhile, to compensate for the loss of users that are not associated with RISs, the proposed FSL scheme allocates appropriate wireless bandwidth for each user to increase the received local FL model parameters hence improving the prediction accuracy.

\begin{figure}[t]
  \centering
  \captionsetup{font={footnotesize }}
  \includegraphics[width=3.1in, height=2.55in]{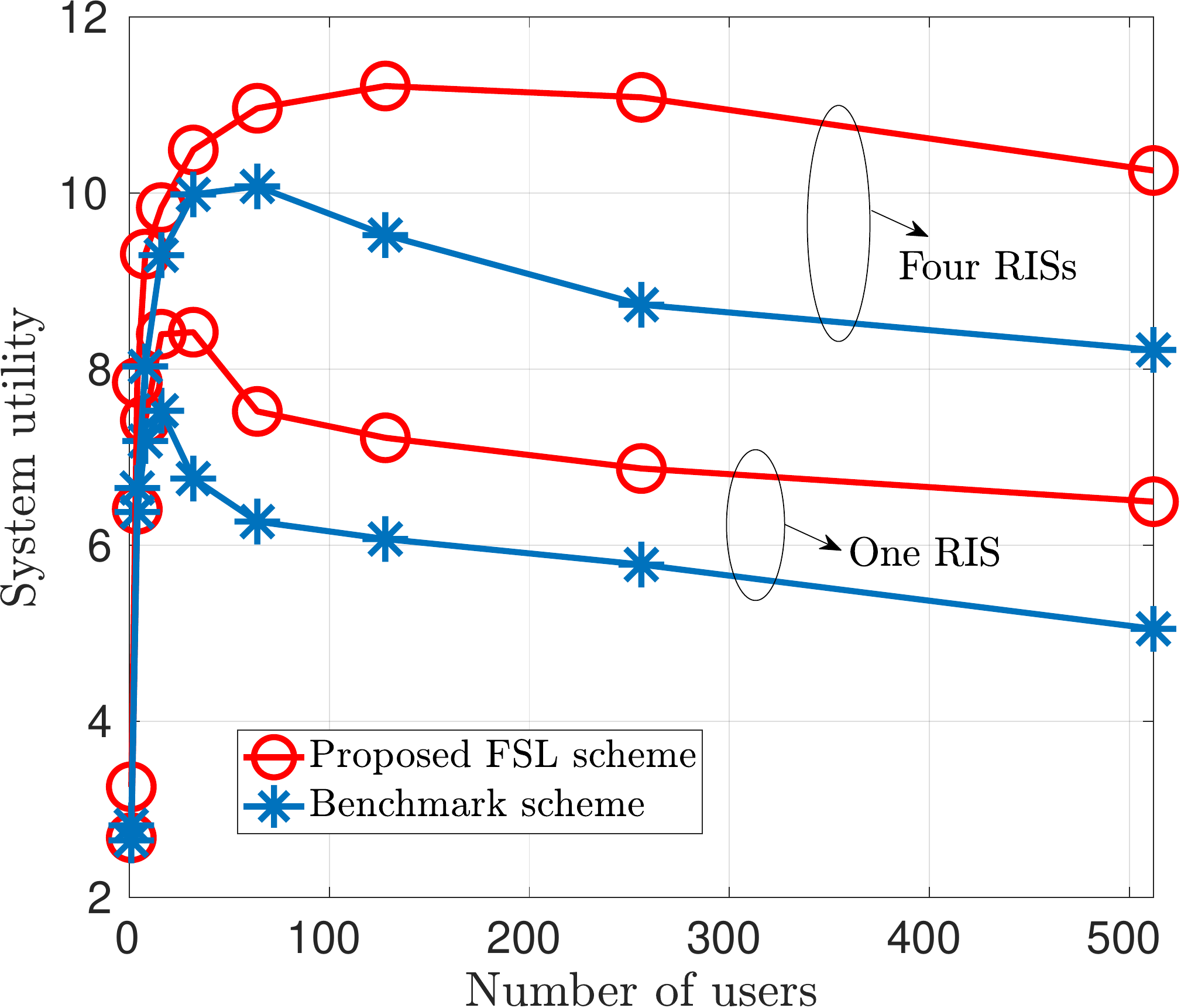}
  \caption{The system utility versus the number of users, where the number of elements of each RIS is $128$.}
  \label{utility_users}
\end{figure}
\begin{figure}[t]
  \centering
  \captionsetup{font={footnotesize }}
  \includegraphics[width=3.1in, height=2.55in]{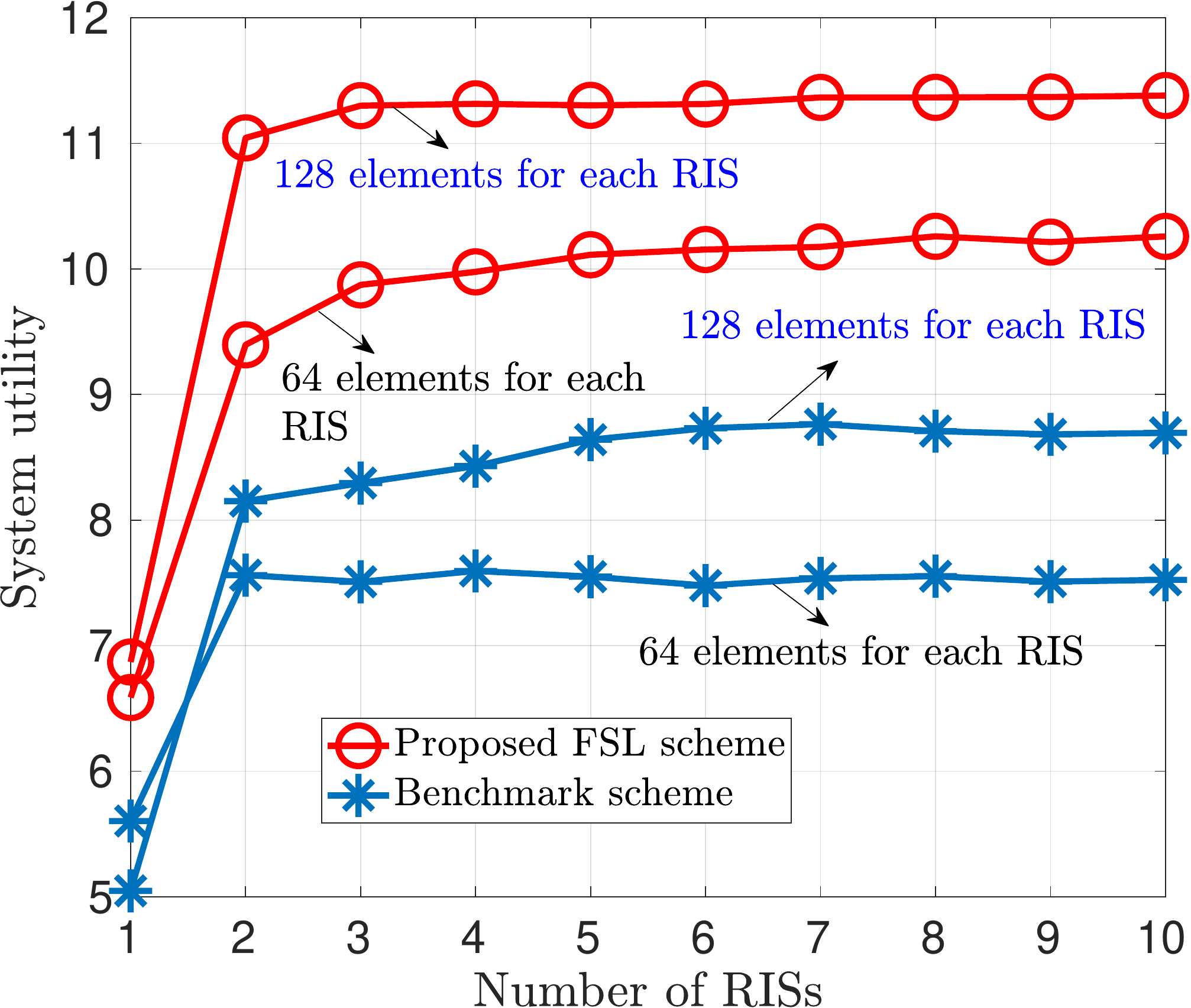}
  \caption{The system utility versus the number of RISs, where the number of users is $256$.}
  \label{utility_RISs}
\end{figure}
\vspace{-3mm}
\subsection{FL System Utility}
\subsubsection{System utility versus number of users}
The FL system utility of the two considered schemes versus the number of users is shown in Fig.~\ref{utility_users},  where the number of elements of each RIS is $128$. From this figure, we can see that, as the number of users increases, the system utility of all considered schemes increases first and then decays. This is because as the number of users increases, the number of local FL model parameters received at the edge server (denoted as ${\rm Q}=\sum_{m=1}^{M} Q(m), \ \forall m \in \cal M$) used for FL aggregation increases with relatively low communication delay. As the number of users continues to increase by contrast, the
system utility decreases slowly. This is mainly due to the fact that the allocated wireless bandwidth for each user decreases and the communication delay becomes large. However, there is no corresponding significant improvement on the number of received local FL model parameters at the edge server. Fig.~\ref{utility_users} also shows that, for a network with $256$ users, the system utility gain achieved by the proposed FSL scheme with four RISs is up to $60\%$ better than that with one RIS. This is because, for a dense users deployment network, the wireless links with poor channel quality can be improved significantly by deploying a large number of RISs.

\begin{figure}[t] 
\centering
\captionsetup{font={footnotesize }}
\subfigure[One RIS]{
\includegraphics[width=1.65in,height=1.15in]{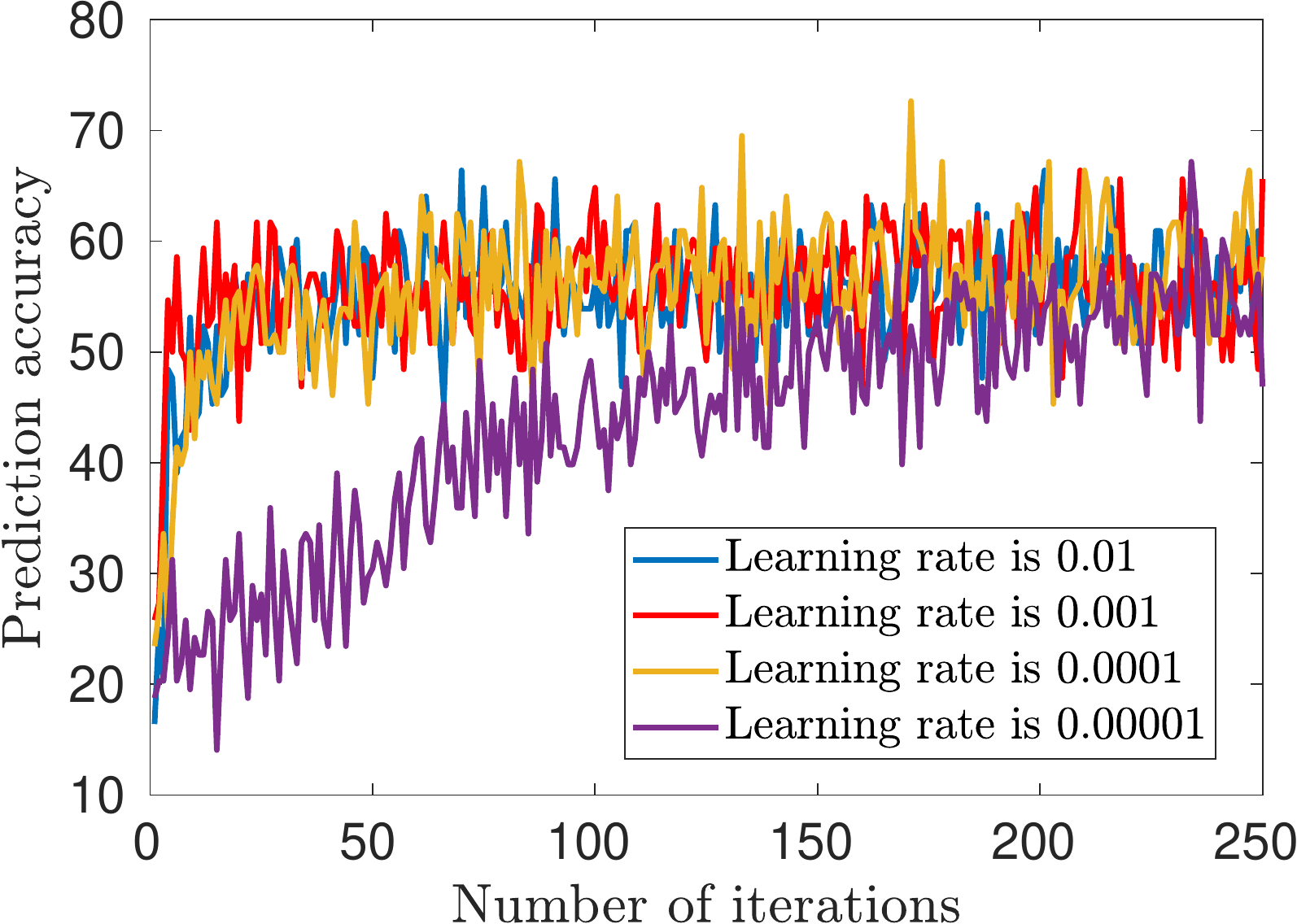}}
\subfigure[One RIS]{
\includegraphics[width=1.65in,height=1.15in]{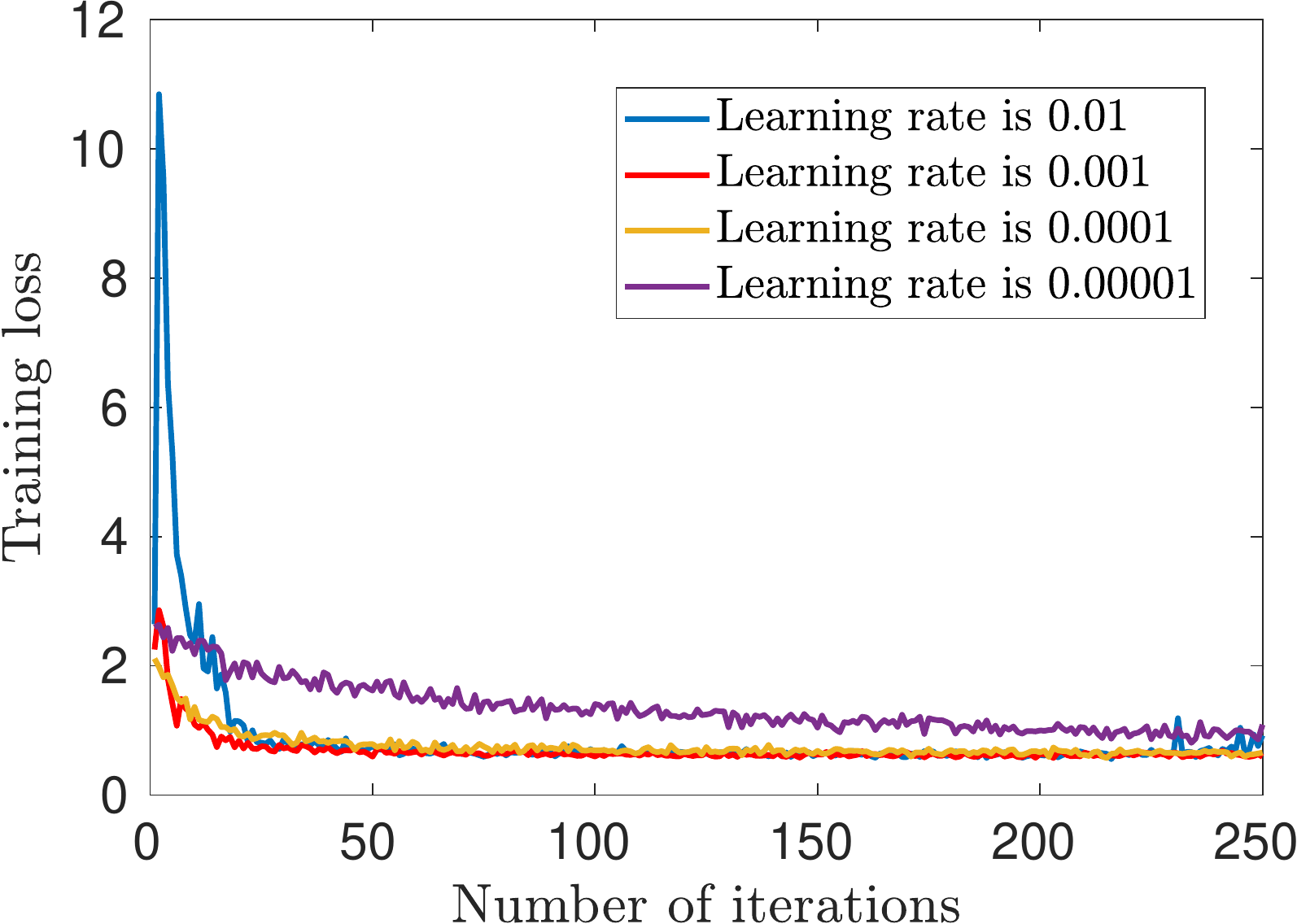}}
\subfigure[Two RISs]{
\includegraphics[width=1.65in,height=1.15in]{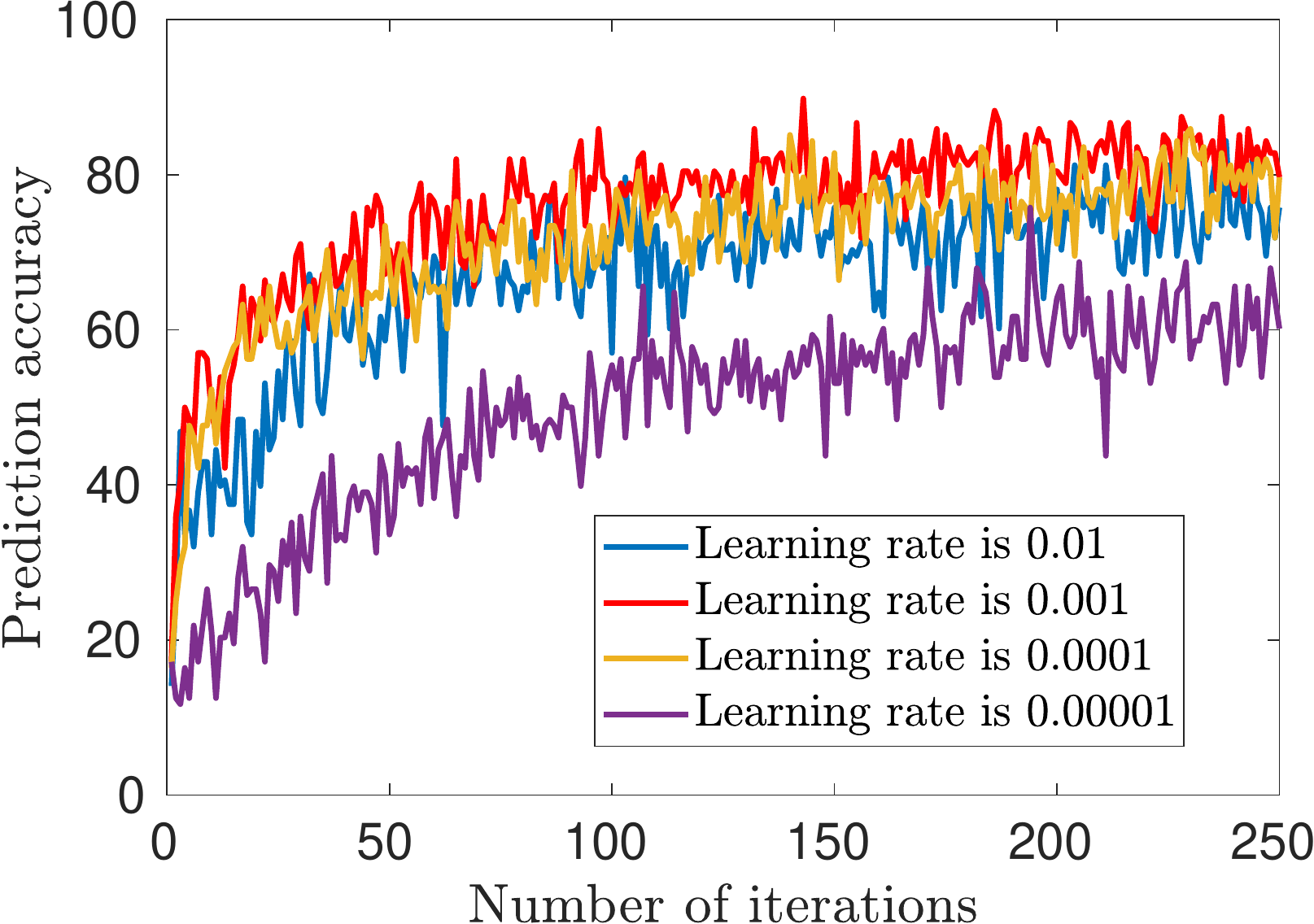}}
\subfigure[Two RISs]{
\includegraphics[width=1.65in,height=1.15in]{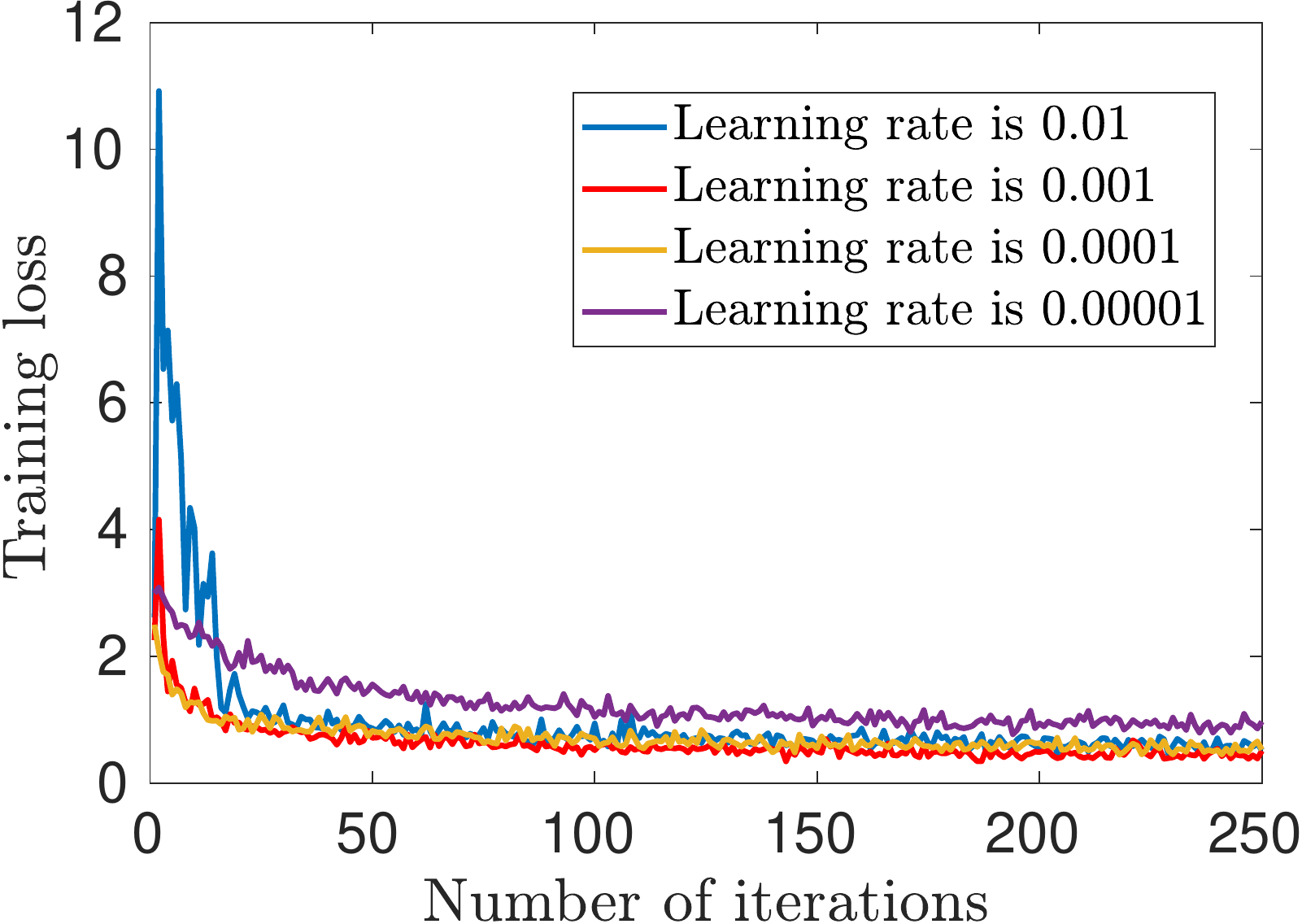}}
\subfigure[Four RISs]{
\includegraphics[width=1.65in,height=1.15in]{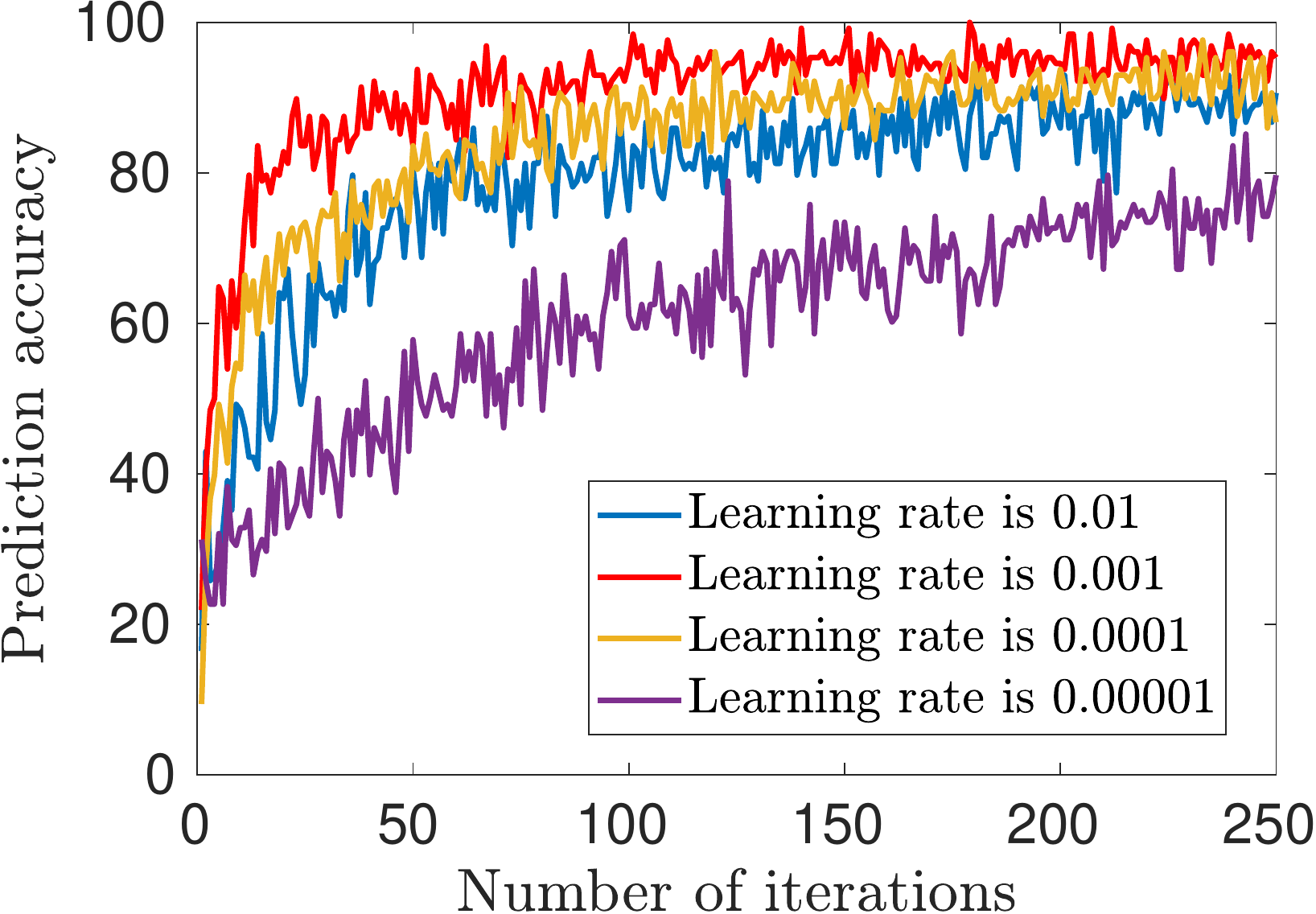}}
\subfigure[Four RISs]{
\includegraphics[width=1.65in,height=1.15in]{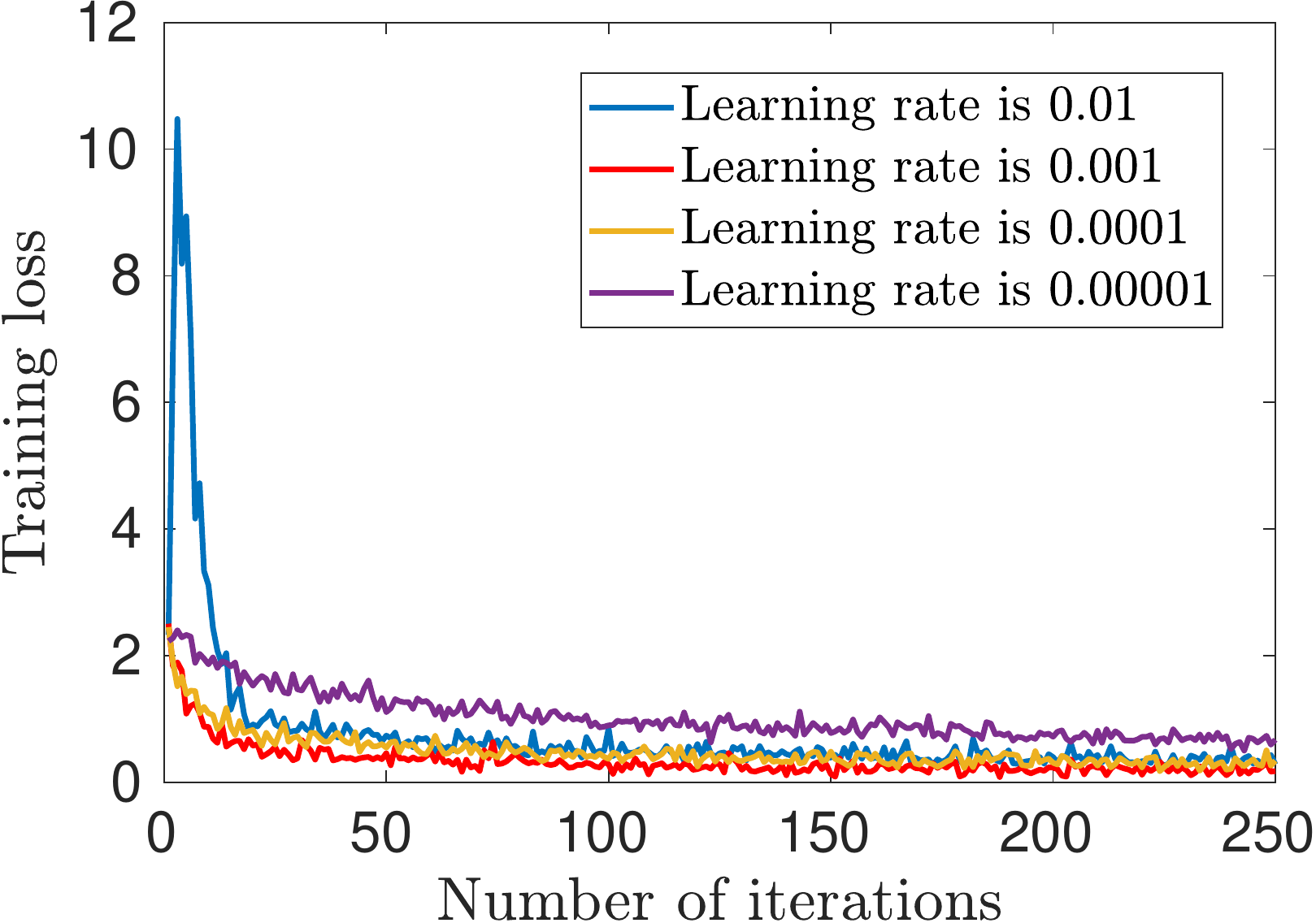}}
\caption{Prediction accuracy and training loss versus number of iterations under different learning rates.}
\label{learning_rate}
\end{figure} 

\subsubsection{System utility versus number of RISs}
The FL system utility of the two considered schemes versus the number of RISs is shown in Fig.~\ref{utility_RISs},  where the number of users is $256$.
One can observe from Fig.~\ref{utility_RISs} that the system utility of all the considered schemes first increases with the number of RISs and then remains stable. As the number of reflecting elements of each RIS increases, the system utility can be further improved.
This is due to the fact that, with more RISs with a larger number of reflecting elements, the RIS can generate more accurate passive reflective beamforming for the incident signals, thereby effectively improving the propagation conditions. Therefore, the edge server is capable of allowing more users to participate in FL so as to improve the system utility. In addition, since the benchmark scheme has a relatively high probability of failing to choose the optimal user-RIS association, we observe that the gap between the benchmark and the proposed FSL schemes becomes significant, especially when the number of RISs is large.

\subsection{Impact of Learning Rate}
The prediction accuracy and the training loss of the proposed FSL framework versus the number of iterations under different learning rates (i.e., $0.01$, $0.001$, $0.0001$, and $0.00001$) is illustrated in Figs.~\ref{learning_rate}(a)-(f), where the considered network is with only one RIS in (a) and (b), with two RISs in (c) and (d), and with four RISs in (e) and (f), respectively. To minimize the loss, the stochastic gradient descent with different learning rates is adopted at each iteration, where the size of each mini-batch is $256$. It is observed that the FL performance of the proposed FSL scheme is influenced by the learning rates.
Specifically, with only one RIS available, Figs.~\ref{learning_rate}(a)-(b) show a close performance of the FSL scheme with learning rates $0.01$, $0.001$ and $0.0001$, while the smallest learning rate $0.00001$ has the worse performance. Interestingly, as the number of RISs increases, e.g., in Figs.~\ref{learning_rate}(c)-(d) with two RISs, the FSL scheme with learning rate $0.001$ gradually outperforms the other three cases. This phenomenon is further demonstrated in Figs.~\ref{learning_rate}(e)-(f), where a network with four RISs is considered.
This is because a too large learning rate
increases the oscillation while a too small learning rate leads to over-fitting.
Therefore, an appropriate learning rate needs to be carefully selected in the proposed FSL framework  based on the number of RISs.

\section{Conclusion}\label{conclusion}
In this article, we explored the symbiotic interplay between federated learning and RISs. To achieve wireless spectrum learning in RISs-aided wireless edge networks, we proposed a novel FSL framework by jointly optimizing the phase shifts, user-RIS association and wireless bandwidth allocation. Simulation results demonstrated the advantages of the proposed FSL framework in terms of spectrum prediction accuracy and  system utility. The proposed FSL framework can be further explored to empower conventional RIS-aided networks with distributively-yet-intelligently `{think-and-decide}' mechanisms. 

 \begin{appendices} 
 \label{App}       

 \section{Proof of Theorem 2}       
 
 Suppose that the FL model is deployed at a central monitor (e.g., the edge server) for the centralized spectrum sensing, the achieved inference accuracy is obtained as $\epsilon$. In the proposed FSL framework, each RIS controller is deployed with a FL model, which cooperatively infers the spectrum information and sends to the edge server for fusion.
Let $K$ denote the number of RISs, the detection accuracy via cooperative spectrum sensing can be obtained as 
$\eta\!=\!1-(1-\epsilon)^K$~\cite{cooper_sensing}. So the performance improvement brought by RISs in the first stage is achieved as 
\begin{equation} 
\chi_1=\frac{\eta}{\epsilon}=\frac{1-(1-\epsilon)^K}{\epsilon}.
\end{equation}

In the conventional FL system, the probability that the local model of the $m$th user can be correctly received at the BS equals to the probability of $\gamma_m^{d} \geq \gamma_{T}$, i.e., 
\begin{equation} \label{p}
p_c={\rm{Pr}}\left \{\gamma_m^{d} \geq \gamma_{T} \right \}=\int_{\gamma_{T}}^{\infty } e^{-x}dx=e^{-\gamma_{T}},
\end{equation}
where $\gamma_m^{d}=\frac{{p_m} \left |h_{d,m} \right |^2 }{{\beta}_m B {\cal N}_0}$ denotes the received SNR of the $m$th user via the direct link.

In the proposed FSL framework, with the aid of RISs' reflection, the received SNR of the $m$th user is obtained as
 \begin{equation} \label{sinr_ris}
\gamma_m= \frac{{p_m} \left |h_{d,m}+ \sum_{k=1}^{K} r_{m,k} \textit{\textbf g}_{k} \mathbf{\Phi}_k \textit{\textbf h}_{m,k}  \right |^2 }{{\beta}_m B {\cal N}_0}.
\end{equation}

By rewriting (\ref{sinr_ris}), we have 
 \begin{equation} \label{sinr_ris1}
 \begin{aligned}
\sqrt{\gamma_m}&= \frac{{\sqrt{p_m}} }{\sqrt{{\beta}_m B {\cal N}_0}} \left |h_{d,m}+ \sum_{k=1}^{K} r_{m,k} \textit{\textbf g}_{k} \mathbf{\Phi}_k \textit{\textbf h}_{m,k}  \right | \\
& \leq \frac{{\sqrt{p_m}} }{\sqrt{{\beta}_m B {\cal N}_0}} \left( \left |h_{d,m}\right |+\left |\sum_{k=1}^{K} r_{m,k} \textit{\textbf g}_{k} \mathbf{\Phi}_k \textit{\textbf h}_{m,k}  \right | \right) \\ &=\sqrt{\gamma_m^d} + \frac{{\sqrt{p_m}} }{\sqrt{{\beta}_m B {\cal N}_0}} \left |\sum_{k=1}^{K} r_{m,k} \textit{\textbf g}_{k} \mathbf{\Phi}_k \textit{\textbf h}_{m,k}  \right |.
\end{aligned}
\end{equation}

Based on (\ref{sinr_ris1}), we have
 \begin{equation} \label{sinr_ris2}
 \begin{aligned}
\gamma_m&\leq \left(\sqrt{\gamma_m^d} + \frac{{\sqrt{p_m}} }{\sqrt{{\beta}_m B {\cal N}_0}} \left |\sum_{k=1}^{K} r_{m,k} \textit{\textbf g}_{k} \mathbf{\Phi}_k \textit{\textbf h}_{m,k}  \right | \right)^2 \\
&\triangleq \gamma_m^d + \kappa,
\end{aligned}
\end{equation}
where
\begin{equation} \nonumber
 \begin{aligned}
\kappa &\triangleq\frac{{{p_m}} }{{\beta}_m B {\cal N}_0} \left |\sum_{k=1}^{K} r_{m,k} \textit{\textbf g}_{k} \mathbf{\Phi}_k \textit{\textbf h}_{m,k}  \right |^2 \\
&+   \frac{2\sqrt{\gamma_m^d p_m} }{\sqrt{{\beta}_m B {\cal N}_0}} \left |\sum_{k=1}^{K} r_{m,k} \textit{\textbf g}_{k} \mathbf{\Phi}_k \textit{\textbf h}_{m,k}  \right |.
  \end{aligned}
\end{equation}

Therefore, the probability that the local models of the $m$th user can be correctly received at the BS is calculated as
\begin{equation} \label{p}
\begin{aligned}
p&={\rm{Pr}}\left \{{\gamma_m} \geq {\gamma_{T}} \right \} \leq  {\rm{Pr}}\left \{\gamma_m^d + \kappa \geq \gamma_{T} \right \} \\
&={\rm{Pr}}\left \{\gamma_m^{d} \geq \gamma_{T}- \kappa\right \}=\int_{\gamma_{T}\!-\! \kappa}^{\infty } e^{-x}dx=e^{-\gamma_{T}+ \kappa}.
\end{aligned}
\end{equation}

Therefore, the performance improvement in the second stage is limited by
\begin{equation} 
\chi_2=\frac{p}{p_c}\leq e^{\kappa}.
\end{equation}
 
 \vspace{2mm}

 \end{appendices}


\end{document}